\documentclass{lmcs}

\newcommand{\myenvalias}[2]{\newenvironment{#1}{\begin{#2}}{\end{#2}}}
\myenvalias{theorem}{thm}
\myenvalias{lemma}{lem}
\myenvalias{corollary}{cor}
\myenvalias{proposition}{prop}
\myenvalias{definition}{defi}
\myenvalias{notation}{nota}
\myenvalias{example}{exa}
\myenvalias{remark}{rem}
\myenvalias{claim}{clm}

\usepackage{amsmath,amssymb,amsfonts,mathtools}
\usepackage{mathrsfs}
\usepackage[dvipsnames]{xcolor}
\usepackage{tikz}
\usetikzlibrary{automata, positioning, arrows.meta, fit, calc, decorations.pathmorphing, shapes, patterns}
\usepackage{hyperref}

\usepackage[capitalise,noabbrev,nameinlink]{cleveref}
\crefname{thm}{Theorem}{Theorems}
\crefname{lem}{Lemma}{Lemmas}
\crefname{cor}{Corollary}{Corollaries}
\crefname{prop}{Proposition}{Propositions}
\crefname{defi}{Definition}{Definitions}
\crefname{nota}{Notation}{Notations}
\crefname{exa}{Example}{Examples}
\crefname{rem}{Remark}{Remarks}
\crefname{clm}{Claim}{Claims}
\crefname{fact}{Fact}{Facts}

\usepackage{stmaryrd}
\usepackage{forest}

\usepackage[linesnumbered, vlined, ruled]{algorithm2e}
\SetKwInput{Input}{Input}
\SetKwInput{Output}{Output}
\crefname{algocf}{Algorithm}{Algorithms}

\newcommand{\N}{\mathbb{N}}
\newcommand{\Z}{\mathbb{Z}}
\newcommand{\Q}{\mathbb{Q}}

\newcommand{\sem}[1]{\llbracket{#1}\rrbracket}

\newenvironment{decisionproblem}[3]{
  \begin{quote}
  \noindent\textsc{#1}\\
  \textbf{Instance:} #2\\
  \textbf{Question:} #3
  \end{quote}
}

\renewcommand{\vec}[1]{\mathbf{#1}}
\newcommand{\card}[1]{ \lvert #1 \rvert }
\newcommand{\restr}[2]{{#1}_{\mid #2}}

\newcommand{\BB}{\mathcal{B}}
\newcommand{\EE}{\mathcal{E}}
\newcommand{\VV}{\mathcal{V}}
\newcommand{\ZZ}{\mathcal{Z}}
\newcommand{\TT}{\mathcal{T}}
\newcommand{\SpS}{\mathcal{S}}
\newcommand{\RR}{\mathcal{R}}

\newcommand{\Sect}{\boldsymbol{\mathrm{Sect}}}
\newcommand{\Reach}{\boldsymbol{\mathrm{Reach}}}
\newcommand{\Extend}{\mathrm{Extend}}
\newcommand{\arity}{\mathrm{ar}}

\newcommand{\target}{\boldsymbol{\mathrm{tgt}}}

\newcommand{\tgt}[1]{\boldsymbol{tgt}(#1)} %
\newcommand{\ra}[1]{\operatorname{ar}(#1)}%

\newcommand{\source}{\boldsymbol{\mathrm{src}}}
\newcommand{\counters}{\boldsymbol{\mathrm{cnt}}}

\newcommand{\Runs}{\boldsymbol{\mathrm{Runs}}}

\newcommand{\Rule}{\text{rule}}
\DeclareMathOperator{\bd}{bd}

\DeclareMathOperator{\sn}{sn}

\begin{document}

\title{Bridging the Gap Between Plain VASS and Branching VASS}

\thanks{%
This work was supported by the grant ANR-25-CE48-6933
of the French National Research Agency
(project CoqoPetri).
Clotilde Bizière was supported by the Polish National Science Centre
under SONATA BIS-12 grant no. 2022/46/E/ST6/00230.
}

\author[C.~Bizière]{Clotilde Bizière\lmcsorcid{0009-0003-6469-1170}}[a,b]
\author[J.~Leroux]{Jérôme Leroux\lmcsorcid{0000-0002-7214-9467}}[a]
\author[G.~Sutre]{Grégoire Sutre\lmcsorcid{0009-0004-3839-0005}}[a]

\address{University of Bordeaux, CNRS, Bordeaux INP, LaBRI, UMR 5800, F-33400 Talence, France}

\address{University of Warsaw, Poland}

\begin{abstract}
  Vectors addition systems with states (VASS), a model equivalent to Petri nets, are finite-state machines with finitely many counters ranging over the natural numbers. The decidable reachability problem for VASS has many applications in logic, automata, and verification.
  In this paper we study the reachability problem for BVASS, a branching
  generalization of VASS. We show that BVASS reachability sets are very similar to VASS reachability sets,
  namely that they are sections of VASS. Our proof relies on a new well-quasi-order (wqo) on BVASS runs that generalizes the
  well-known wqo on VASS runs. By leveraging an amalgamation property, we prove that every BVASS run can be transformed into
  an equivalent one of bounded branching complexity. This allows us to derive several results on the geometry of
  BVASS reachability sets. As an application we obtain that reachability sets of
  $5$-dimensional BVAS are effectively semilinear, as is the case for $5$-dimensional VAS.
\end{abstract}

\maketitle

\section{Introduction} \label{sec:intro}

\paragraph{Context}
Vectors addition systems with states (VASS), a model equivalent to Petri nets, are finite-state machines with finitely many counters ranging over the natural numbers.
Operations on counters are limited to increment and guarded decrement (a counter can be decremented only if it remains non-negative).
A central decision problem for VASS is reachability: whether there exists a run from an initial configuration to a final one.
This problem was shown decidable more than forty years ago~\cite{Mayr84} and was recently revisited~\cite{DBLP:conf/lics/LerouxS19,DBLP:journals/jacm/CzerwinskiLLLM21,DBLP:conf/focs/Leroux21}.
The decidability of the reachability problem for VASS is the cornerstone of many decidability results in logic, automata, and verification~\cite{DBLP:journals/siglog/Schmitz16}.

\smallskip

Several VASS extensions have been proposed to increase the expressive power of the model, most notably VASS with nested zero tests~\cite{REINHARDT2008239}, pushdown VASS~\cite{DBLP:conf/fsttcs/AtigG11,Lazic2013}, unordered data nets~\cite{DBLP:journals/fuin/LazicNORW08}, and branching VASS~\cite{GGS04,DBLP:journals/dmtcs/VermaG05}. 
The reachability problem is decidable for VASS with nested zero tests~\cite{REINHARDT2008239,CGL25} and, as shown very recently, for pushdown VASS~\cite{guttenberg_et_al:LIPIcs.LICS.2026.53}. It remains open for unordered data nets and branching VASS.

\smallskip

In this paper, we investigate the reachability problem for branching VASS (shortly called BVASS in the sequel), a branching generalisation of VASS. More precisely, BVASS extend VASS with special branching transitions that merge configurations (by summing their vectors). Thus runs for BVASS are trees of configurations (whereas they are sequences of configurations for plain VASS). The BVASS model has gained a lot of interest recently due to strong links with several fields in computer science such as cryptographic protocols~\cite{DBLP:journals/dmtcs/VermaG05}, linear logic~\cite{GGS04,LazicS15}, recursively parallel programs~\cite{DBLP:journals/toplas/BouajjaniE13}, timed pushdown systems~\cite{DBLP:conf/lics/ClementeLLM17}, computational linguistic~\cite{DBLP:conf/acl/Rambow94,DBLP:conf/acl/Schmitz10}, game semantics~\cite{DBLP:conf/esop/Cotton-BarrattM17}, equational tree automata~\cite{DBLP:conf/csl/Ohsaki01,DBLP:conf/fossacs/Lugiez03} and data logics~\cite{DBLP:journals/corr/JacquemardSD16,DBLP:conf/pods/BojanczykDMSS06}.
As mentioned before, the reachability problem is still open in arbitrary dimension for BVASS. In dimensions one and two, the reachability problem is decidable~\cite{DBLP:conf/mfcs/BiziereHLS25} and the exact complexity is known in dimension one~\cite{GollerHLT16,DBLP:conf/icalp/FigueiraLLMS17}.

\paragraph{Contributions}
We introduce a well-quasi-order (wqo) on the set of runs of a BVASS that generalises the well-known wqo on VASS runs. As in the case of VASS, this wqo satisfies the \emph{amalgamation property}. Using this property, we prove that reachable configurations are reachable by runs of bounded branching complexity. The idea is to transfer complex computations from one node to a descendant of a sibling node. From this simple form of runs, we deduce several results.
\begin{itemize}
\item Reachability sets of BVASS are VASS sections, \emph{i.e.}, projections on a subset of counters of VASS reachability sets intersected with semilinear sets. From this characterisation, we derive that BVASS reachability sets are almost semilinear, as for VASS. Moreover, we prove that when the reachability set of a BVASS is semilinear then it is effectively computable, again as for VASS.
\item The so-called Strahler-bounded reachability problem is decidable. This problem asks whether a given configuration is reachable by a run whose Strahler number is at most a given bound. We solve that problem by reduction to the reachability problem of VASS with nested zero tests~\cite{REINHARDT2008239}. The Strahler-bounded reachability problem can be seen as a precise under-approximation of the BVASS reachability problem since we prove that for every BVASS, there exists a uniform bound such that every reachable configuration is reached by a run whose Strahler number is at most that bound.
\item We then focus on small dimensions and provide an iterative fix-point algorithm based on semilinear sets. The termination of this algorithm is guaranteed by the previously shown uniform bound. This iterative algorithm shows that reachability sets are effectively semilinear for $2$-dimensional BVASS and $5$-dimensional BVAS (BVAS are BVASS with a single state). For $2$-BVASS, another algorithm computing the semilinear reachability set was given in~\cite{DBLP:conf/mfcs/BiziereHLS25}, but the proof of termination was significantly longer and more involved.
\end{itemize}

\paragraph{Related Work}
In general dimension, the coverability problem (a weak version of the reachability problem) and the boundedness problem (the question whether the reachability set is finite) are decidable for BVASS~\cite{DBLP:journals/dmtcs/VermaG05}, and their precise complexity is known~\cite{DBLP:journals/jcss/DemriJLL13,LazicS15}. The \emph{Strahler number} was introduced by Horton and Strahler in the 1950s in the context of hydrology, and reinvented many times since then, see~\cite{Strahler}. In~\cite{DBLP:conf/fsttcs/AtigG11}, it is shown that the reachability problem for VASS along finite-index context-free languages is decidable by reduction to the reachability problem for VASS with nested zero tests. As observed in~\cite{Strahler}, for a context-free grammar in Chomsky normal form, the index of a derivation tree is equal to its Strahler number minus one.

\smallskip
This paper is an extended version of our FOSSACS'26 paper~\cite{DBLP:conf/fossacs/BiziereLS26} with additional examples, explanations, and detailed proofs. Whereas the notion of trees for defining BVASS runs is based on graphs in the conference paper, in this extended version, trees are defined inductively in a term-like fashion. Thanks to this point of view, the well-quasi-order on BVASS runs can be defined in a simpler way, many proofs can be easily done inductively on the structure of trees, and last but not least, we avoid the burden of defining the complex notion of "hereditarily positive family of contexts" introduced in the original version of the paper to prove amalgamation results. The two notions are slightly different since children are ordered in the term-like definition, but this has no impact for the reachability problem.

\section{Preliminaries} \label{sec:prelim}

\subsection{Generalities}

We let $\Z$, $\N$ and $\Q$ denote the usual sets of integers, natural numbers and rational numbers, respectively. We write $\Q_{>0}$ (resp., $\Q_{\geq 0}$) the set of positive (resp., nonnegative) rational numbers.
For every $a,b \in \N$, we denote by $[a,b] \coloneqq \{ n \in \N \mid a \le n \le b\}$ the integer interval between $a$ and $b$. The cardinality of a set $S$ is written $\card{S}$.

\paragraph{Vectors}
Given a finite set $I$ of indices and a set $\mathbb S \in \{\N,\Z,\Q,\Q_{>0},\Q_{\geq 0}\}$, an \emph{$I$-vector over $\mathbb S$} is a function $I \to \mathbb S$.
Most often, we take $I = [1,d]$, yielding the usual set $\mathbb S^d$ of \emph{$d$-dimensional vectors}.
Vectors are written in boldface and, for $i \in I$, the \emph{component $i$} of the vector $\vec v$ is denoted $\vec v (i)$.
The restriction of an $I$-vector $\vec v$ to a subset $J \subseteq I$ is denoted $\restr{\vec v}{J}$.
The zero $I$-vector is denoted $\vec 0_I$, $\vec 0_d$ when $I = [1,d]$ or $\vec 0$ when $I$ is clear from context.
The product ordering of the usual ordering on $\mathbb{S}$ is denoted $\le$, that is $\vec u \le \vec v$ if and only if for all $i \in I$, $\vec u(i) \le \vec v(i)$.

\paragraph{Semilinear Sets}
A \emph{linear set} of $\N^d$ is a set of the form $\vec{L}=\vec{b}+\vec{P}$ where $\vec{b}\in\N^d$ is called the \emph{basis} and $\vec{P}=\{n_1\vec{v}_1+\cdots+n_k\vec{v}_k\mid n_1,\ldots,n_k\in\N\}$ is called the \emph{periodic set spanned} by the vectors $\vec{v}_1,\ldots,\vec{v}_k\in\N^d$ called the \emph{periods}. The pair $\gamma=(\vec{b},\vec{V})$ is called a \emph{presentation} of the linear set $\vec{L}$, and we denote by $\sem{\gamma}$ the linear set $\vec{L}$ presented by $\gamma$. A \emph{semilinear set} of $\N^d$ is a set of the form $\vec{S}=\vec{L}_1\cup\ldots\cup \vec{L}_k$ where $\vec{L}_j$ is a linear set for every $1\leq j\leq k$. A finite set $\Gamma=\{\gamma_1,\ldots,\gamma_k\}$ where $\gamma_j$ is a presentation of the linear set $\vec{L}_j$ for every $1\leq j\leq k$ is called a \emph{presentation} of $\vec{S}$. In that case, we denote by $\sem{\Gamma}$ the semilinear set $\vec{S}$ presented by $\Gamma$. A semilinear set is said to be \emph{effectively computable} if a presentation of the semilinear set is \emph{computable}.

\paragraph{Rooted labelled trees}
A \emph{(finite rooted directed labelled) tree} $\alpha$ on a set $\Lambda$ is defined inductively as a pair $(\lambda,(\alpha_1,\ldots, \alpha_k))$ where $\lambda\in \Lambda$, $k\in\N$, and $\alpha_1,\ldots,\alpha_k$ is a (possibly empty) finite sequence of trees on $\Lambda$. Such a pair is simply denoted by $\lambda\langle \alpha_1,\ldots,\alpha_k\rangle$. The trees $\alpha_1,\ldots,\alpha_k$ are called the \emph{immediate subtrees} of $\alpha$ and are denoted $\alpha[1],\ldots,\alpha[k]$, while the integer $k$ is called the \emph{arity} of $\alpha$ and is written $\ra{\alpha}$. We denote by $\TT(\Lambda)$ the set of trees on $\Lambda$.

\begin{remark}
  In this paper, operations on trees are formally defined by structural induction. However, we will provide some intuitions using the usual graph theory viewpoint of trees. In particular, we will use the classical concepts of nodes, children, root, leaves, unary nodes and branches of a tree.
\end{remark}

The \emph{size} $\lvert \alpha \rvert$ of a tree $\alpha$ is defined inductively by $\lvert \alpha \rvert \coloneqq 1 + \sum_{i=1}^{\ra{\alpha}} \lvert \alpha[i] \rvert$. Intuitively, $|\alpha|$ is the number of nodes in the tree $\alpha$.
We define the binary relation $\sqsubseteq$ on trees inductively as follows.
A tree $\alpha$ is a \emph{subtree} of a tree $\beta$, denoted $\alpha \sqsubseteq \beta$, if $\alpha = \beta$ or $\alpha \sqsubseteq \beta[i]$ for some $i \in [1,\ra{\beta}]$.
Intuitively, $\alpha\sqsubseteq \beta$ if there exists a node in $\beta$ such that the subtree rooted at this node is equal to $\alpha$.

\paragraph{Branching depth and Strahler number}
We use two measures of the structural complexity of a tree, namely the \emph{branching depth} and the \emph{Strahler number}.  

\smallskip

The \emph{branching depth} $\bd(\alpha)$ of a tree $\alpha$ is defined inductively by $\bd(\alpha) = 0$ if $\ra{\alpha} = 0$, $\bd(\alpha) = \bd(\alpha[1])$ if $\ra{\alpha} = 1$ and $\bd(\alpha) = 1 + \max_{i=1,\ldots,\ra{\alpha}}\bd(\alpha[i])$ if $\ra{\alpha} \ge 2$.
Intuitively, it is the depth of the tree obtained from $\alpha$ by merging every unary node with its child.

\smallskip

The second measure of tree complexity that we use is the \emph{Strahler number}, a less well-known concept with an interesting history.  
It was introduced by Horton and Strahler in the 1950s in hydrology as a tool to quantify the branching complexity of river networks.  
Since then, Strahler numbers have been reinvented many times under different names and applied in a wide range of fields within computer science.  
We refer the reader to the survey~\cite{Strahler} for a detailed account of these developments.  

\smallskip

The Strahler number $\sn(\alpha)$ of a tree $\alpha$ is defined recursively as follows.  
If $\ra{\alpha} = 0$, then $\sn(\alpha) = 0$. 
Otherwise, let  $s_i = \sn(\alpha[i])$ and $s \coloneqq \max\{s_1,\dots,s_{\ra{\alpha}}\}$.  
Then $\alpha$ has Strahler number $s+1$ if $s$ occurs at least twice among the $s_i$, and $s$ otherwise.
An example of a tree where nodes are labelled by the Strahler number of the subtree rooted at this node is given on the left of \cref{fig:Strahler-bd}.

\begin{figure}
  \centering
  \begin{forest}
    [2
    [1 [0] [1 [0] [0] [0]]]
    [1 [0] [0 [0]]]    ]
  \end{forest}
  \hfill
  \begin{tikzpicture}[
    >=Stealth,
    node distance=2.5cm,
    every state/.style={minimum size=8mm},
    auto
  ]

  \node[state] (p) {$p$};
  \node[state,right=of p] (q) {$q$};
  \node[below=0.8cm of p] (start) {};

  \path[->] (p) edge[bend left] node[above] {$(1,0,0)$} (q);
  \path[->] (q) edge node[above] {} (p);

  \path[->,red] (p) edge[bend right=30] coordinate[pos=0.8] (pq1) (q);
  \path[->,red] (p) edge[bend right=50] coordinate[pos=0.8] (pq2) (q);
  \path[-] (pq2) edge[red, bend left] node[red, left, yshift=-0.1cm] {$+$} (pq1);

  \path[->] (p) edge[loop above] node[above] {$(0,1,-1)$} (p);
  \path[->] (q) edge[loop above] node[above] {$(0,-1,2)$} (q);

  \path[->] (start) edge node[left] {$(0,1,0)$} (p);

\end{tikzpicture}
\hfill
\begin{tikzpicture}[%
  tnode/.style = {circle, draw=black, fill=black, inner sep=0.4ex}
]
  \node (n1) [tnode, label={[yshift=2ex]center:{$q(0,1,2)$}}] {};
  \node (n2) [tnode, label={[xshift=6ex]center:{$q(0,2,0)$}}] at ([yshift=-14mm]n1) {};
  \node (n3) [tnode, label={[yshift=-2ex]center:{$p(0,1,0)$}}] at ([xshift=-10mm, yshift=-10mm]n2) {};
  \node (n4) [tnode, label={[yshift=-2ex]center:{$p(0,1,0)$}}] at ([xshift=10mm, yshift=-10mm]n2) {};

  \path[-]
    (n1) edge (n2)
    (n2) edge (n3)
    (n2) edge (n4)
  ;
\end{tikzpicture}

\caption{On the left: a tree where nodes are labelled by their Strahler number. %
On the middle: a $3$-BVASS. On the right: a run of that BVASS.}
\label{fig:Strahler-bd}
\end{figure}
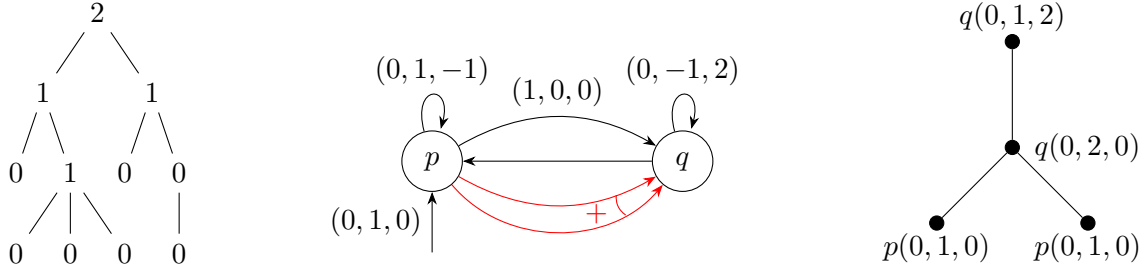

\subsection{Wqo} \label{subsec:wqo} %

This article relies on results about well quasi-orders, but requires no prior familiarity with the topic. We therefore keep the presentation minimal and refer the reader to~\cite{polyWQO} for more background and for complete proofs.  

\smallskip

A \emph{quasi-order} is a pair $(X,\leq)$ where $X$ is a set and $\leq$ is a binary relation on $X$ that is reflexive and transitive.
The \emph{$\leq$-upward-closure} ${\uparrow_{\leq}}(Y)$ of a subset $Y \subseteq X$ is the set $\{x \in X \mid \exists y \in Y, y \leq x\}$.
We say that $Y$ is \emph{$\leq$-upward-closed} when $Y = {\uparrow_{\leq}}(Y)$.
A \emph{$\leq$-basis} of a $\leq$-upward-closed subset $Y \subseteq X$ is a set $B \subseteq Y$ such that
$Y = {\uparrow_{\leq}}(B)$.

\smallskip

A quasi-order $(X,\leq)$ is a \emph{well quasi-order} (wqo) if for every infinite sequence $(x_i)_{i \in \mathbb{N}}$ with $x_i \in X$, there exist indices $i<j$ such that $x_i \leq x_j$. Equivalently, $(X,\leq)$ is a wqo if every $\leq$-upward-closed subset $Y \subseteq X$ admits a finite $\leq$-basis.

\begin{example} \label{ex:wqo}
  $(\mathbb{N},\leq)$ is a wqo. Every finite quasi-order is a wqo (in particular, $(Q,=)$ where $Q$ is a finite set).
\end{example}  

\smallskip

Many natural constructions preserve wqos. We recall below the ones used in this paper. We assume that 
$(X,\leq_X)$ and $(Y, \leq_Y)$ are quasi-orders.

\smallskip

\emph{Order-reflecting functions.}  
If $f:(X,\leq_X)\to(Y,\leq_Y)$ satisfies $f(x) \leq_Y f(x') \implies x \leq_X x'$, and $(Y,\leq_Y)$ is a wqo then $(X, \leq_X)$ is also a wqo. %

\smallskip

\emph{Monotonic functions.}  
If a surjective function $f:(X,\leq_X)\to(Y,\leq_Y)$ satisfies $x \leq_X x' \implies f(x) \leq_Y f(x')$, and $(X,\leq_X)$ is a wqo then $(Y, \leq_Y)$ is also a wqo. %

\smallskip

\emph{Cartesian products.}  
If $(X,\leq_X)$ and $(Y,\leq_Y)$ are wqo, then so is their product $(X \times Y, \leq_{X \times Y})$, where $(x,y) \leq_{X \times Y} (x',y')$ iff $x \leq_X x'$ and $y \leq_Y y'$.  

\smallskip

\emph{Words.}  
If $(X,\leq)$ is a wqo, then so is the set $X^*$ of finite sequences (or words) ordered by the \emph{subsequence embedding}:  
$x_1 \cdots x_m \leq_{X^*} y_1 \cdots y_n$ if there exist indices $1 \leq i_1 < \cdots < i_m \leq n$ such that $x_j \leq y_{i_j}$ for all $j$.  

\smallskip

\emph{Trees.} 
If $(X,\leq)$ is a wqo, then so is the set $\TT(X)$ of finite rooted directed trees labelled by $X$, ordered by \emph{homeomorphic embedding}. The homeomorphic embedding is defined inductively on the structure of a tree by $\alpha \preceq_{\TT} \beta$ where $\alpha\coloneqq x\langle \alpha_1,\ldots,\alpha_m\rangle$ and $\beta\coloneqq y\langle \beta_1,\ldots,\beta_{n}\rangle$ if (1) there exists $i\in [1,n]$ such that $\alpha\preceq_{\TT} \beta_i$, or (2) $x \leq y$ and there exists a sequence $1\leq i_1<\ldots<i_m\leq n$ such that $\alpha_j\preceq_{\TT} \beta_{i_j}$ for all $j$.

\subsection{VASS and extensions}

\paragraph{VASS}
A \emph{$d$-dimensional VASS} ($d$-VASS) is a triple $\VV \coloneqq (Q, \Delta, \vec S)$ where $Q$ is a nonempty finite set of \emph{states}, $\Delta \subseteq Q \times \Z^d \times Q$ is a finite set of \emph{transitions}, and $\vec S \subseteq Q \times \N^d$ is a finite set of initial configurations. The \emph{configurations} of $\VV$ are pairs of the form $\vec c \coloneqq (q, \vec x)$, usually written $q(\vec x)$, with $q \in Q$ and $\vec x \in \N^d$. We extend the product ordering of $\N^d$ to configurations by writing $p(\vec x) \le q(\vec y)$ when $p = q$ and $\vec x \le \vec y$.
Intuitively, a VASS is an automaton equipped with counters which can be incremented and decremented. Counters must always remain nonnegative, and any transition that would make a counter negative is disabled.

\smallskip

A \emph{run} of $\VV$ is a finite sequence $\rho \coloneqq q_1(\vec x_1) q_2(\vec x_2) \cdots q_k(\vec x_k)$ of configurations such that the first configuration $q_1(\vec x_1)$, called \emph{source} and denoted $\source(\rho)$, belongs to $\vec S$, and every other configuration is obtained from the previous one by applying a transition, \emph{i.e.}, $(q_i, \vec x_{i+1} - \vec x_i, q_{i+1}) \in \Delta$ for each $i \in [1,k-1]$. The last configuration $q_k(\vec x_k)$ is called the \emph{target} and denoted $\target(\rho)$.
The set of runs of $\VV$ is denoted $\Runs(\VV)$.
The \emph{reachability set} $\Reach(\VV)$ is the set of targets of runs of $\VV$, \emph{i.e.}, $\Reach(\VV) \coloneqq \{ \target(\rho) \mid \rho \in \Runs(\VV) \}$.
The \emph{sections} of $\VV$ are the sets obtained from $\Reach(\VV)$ by intersecting with a semilinear set and projecting. In this article, we only consider sections of the form $\Sect_I^{Q'} \coloneqq \{ q(\restr{\vec x}{I}) \mid q \in Q', \restr{\vec x}{[1,d] \setminus I} = \vec 0, q(\vec x) \in \Reach(\VV) \}$ where $Q' \subseteq Q$ and $I \subseteq [1,d]$. 
The \emph{reachability problem} asks, given a VASS $\VV$ and a configuration $\vec c$, whether $\vec c \in \Reach(\VV)$.

\paragraph{BVASS}
A \emph{$d$-dimensional branching VASS} ($d$-BVASS) is a pair $\BB \coloneqq (Q, \Delta)$ where $Q$ is a finite set of \emph{states} and $\Delta \subseteq Q^* \times \Z^d \times Q$ is a finite set of \emph{transition rules}. The \emph{arity} $\arity(\delta)$ of a transition rule $\delta \coloneqq (p_1\cdots p_k, \vec a, q)$ is $k$. We say that $\delta$ is \emph{unary} if $\arity(\delta) = 1$ and that $\delta$ is \emph{branching} if $\arity(\delta) \geq 2$. As for VASS, we call \emph{configurations} of $\VV$ the pairs $(q, \vec x) \in Q \times \N^d$, and we write them $\vec c \coloneqq q(\vec x)$. Configurations $q(\vec a)$ where $(\varepsilon, \vec a, q) \in \Delta$ are called \emph{initial configurations}.

\smallskip
\emph{Runs} of $\BB$ are trees with labels in $Q \times \N^d$ satisfying a certain property defined inductively. A tree $\alpha \coloneqq q(\vec x) \langle \alpha_1,\ldots,\alpha_k \rangle$ is a run with \emph{target} $\target(\alpha) \coloneqq q(\vec x)$ if all the trees $\alpha_i$ with $i \in [1,k]$ are runs whose targets $p_i(\vec x_i)$ satisfy $(p_1 \cdots p_k, \vec x - \sum_{i=1}^k \vec x_i, q) \in \Delta$. 
In other words, every node is obtained from its children by applying some transition rule in $\Delta$. We also define $\Rule(\alpha) \coloneqq (p_1 \cdots p_k, \vec x - \sum_{i=1}^k \vec x_i, q)$ and $\counters(\alpha) \coloneqq \vec x$.
Reachability sets and sections are defined and denoted exactly as in the case of VASS. The BVASS reachability problem is also defined similarly. We also use the notation $\Runs(\BB)$ for the set of runs of $\BB$.

\begin{example}
  A $3$-BVASS is depicted in \cref{fig:Strahler-bd}. With the graph theory viewpoint, the run $q(0,1,2)\langle q(0,2,0)\langle p(0,1,0)\langle\rangle p(0,1,0)\langle\rangle \rangle  \rangle$ of that BVASS is also depicted in that figure. Without the branching transition rule depicted by two red arrows, we would have a VASS.
\end{example}

\smallskip

A set of configurations $\vec I\subseteq Q\times \N^d$ is called an \emph{inductive invariant} for a BVASS $\BB=(Q,\Delta)$ if for every sequence $p_1(\vec{x}_1),\ldots,p_k(\vec{x}_k)\in \vec I$ and every configuration $q(\vec{x})$ such that $(p_1 \cdots p_k,\vec{x}-\sum_{i=1}^k\vec{x}_i,q)\in \Delta$, we have $q(\vec{x})\in \vec I$. Observe that the reachability set of $\BB$ is an inductive invariant and is contained in every inductive invariant. In particular a configuration is not reachable for a BVASS, if and only if, there exists an inductive invariant that does not contain that configuration.

\section{A wqo on BVASS runs} \label{sec:wqo-on-runs}

Several results related to the VASS reachability problem have been derived from a geometrical decomposition of the VASS reachability sets into finite unions of almost linear sets~\cite{DBLP:conf/concur/GuttenbergRE23,DBLP:conf/lics/Leroux13}. This decomposition was obtained in~\cite{Turing-100:Vector_Addition_Systems_Reachability} thanks to a wqo on VASS runs that satisfies an amalgamation property. This wqo was first introduced by Jančar in~\cite{JANCARwqo} and independently rediscovered later~\cite{Turing-100:Vector_Addition_Systems_Reachability}. Intuitively, a VASS run $\sigma$ is smaller\footnote{%
  Our informal presentation slightly differs from~\cite{JANCARwqo,Turing-100:Vector_Addition_Systems_Reachability} as they allow comparable runs to have distinct (but comparable) source configurations.
  In our setting, runs always start from an initial configuration,
  of which there are only finitely many,
  so we require the source configurations to be the same for simplicity.
} than a VASS run $\tau$ if $\tau$ is obtained from $\sigma$ by inserting a sequence of cycles $\sigma_1,\ldots,\sigma_k$ that is hereditarily positive, \emph{i.e.}, such that the total effect of $\sigma_1,\ldots,\sigma_i$ is non-negative (on all components) for every $i\in[1,k]$.

\smallskip

In this section we generalise this wqo to BVASS runs and prove an amalgamation property. Our wqo is defined inductively on the tree structure of BVASS runs and, at first sight, seems to differ from the previously mentioned wqo for the special case of VASS. However, as discussed in \cref{rem:wqo} below, our wqo can be informally explained with the notion of contexts, a natural generalisation of VASS cycles.

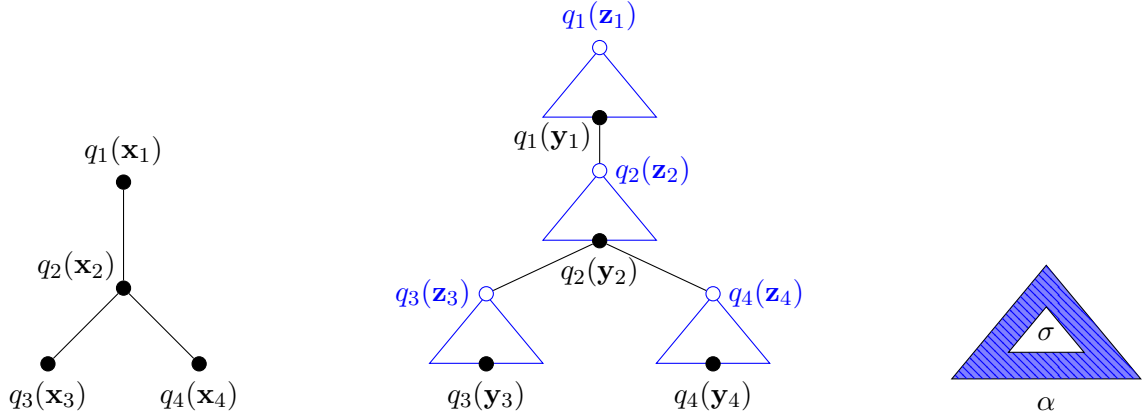
\begin{figure}[t]
  \begin{tikzpicture}[%
  tnode/.style = {circle, draw=black, fill=black, inner sep=0.4ex}
]
  \node (n1) [tnode, label={[yshift=2.5ex]center:{$q_1(\vec{x}_1)$}}] {};
  \node (n2) [tnode, label={[xshift=-3.75ex,yshift=1.75ex]center:{$q_2(\vec{x}_2)$}}] at ([yshift=-14mm]n1) {};
  \node (n3) [tnode, label={[yshift=-2.5ex]center:{$q_3(\vec{x}_3)$}}] at ([xshift=-10mm, yshift=-10mm]n2) {};
  \node (n4) [tnode, label={[yshift=-2.5ex]center:{$q_4(\vec{x}_4)$}}] at ([xshift=10mm, yshift=-10mm]n2) {};

  \path[-]
    (n1) edge (n2)
    (n2) edge (n3)
    (n2) edge (n4)
  ;
\end{tikzpicture}
\hfill
\begin{tikzpicture}[%
  tnode/.style = {circle, draw=black, fill=black, inner sep=0.4ex},
  cnode/.style = {circle, draw=blue, fill=white, inner sep=0.4ex},
  ctx/.style = {isosceles triangle, shape border rotate=90, draw=blue, isosceles triangle stretches, minimum width=1.5cm, minimum height=0.9cm}
]

  \node (gamma1) [ctx] {};
  \node (m1) [cnode, label={[yshift=2.5ex, blue]center:{$q_1(\vec{z}_1)$}}] at (gamma1.north) {};
  \node (n1) [tnode, label={[xshift=-3.75ex,yshift=-1.5ex]center:{$q_1(\vec{y}_1)$}}] at (gamma1.south) {};

  \node (gamma2) [ctx, anchor=north] at ([yshift=-7mm]n1) {};
  \node (m2) [cnode, label={[xshift=4.25ex, blue]center:{$q_2(\vec{z}_2)$}}] at (gamma2.north) {};
  \node (n2) [tnode, label={[yshift=-2.5ex]center:{$q_2(\vec{y}_2)$}}] at (gamma2.south) {};

  \node (gamma3) [ctx, anchor=north] at ([xshift=-15mm, yshift=-7mm]n2) {};
  \node (m3) [cnode, label={[xshift=-4.25ex, blue]center:{$q_3(\vec{z}_3)$}}] at (gamma3.north) {};
  \node (n3) [tnode, label={[yshift=-2.5ex]center:{$q_3(\vec{y}_3)$}}] at (gamma3.south) {};

  \node (gamma4) [ctx, anchor=north] at ([xshift=15mm, yshift=-7mm]n2) {};
  \node (m4) [cnode, label={[xshift=4.25ex, blue]center:{$q_4(\vec{z}_4)$}}] at (gamma4.north) {};
  \node (n4) [tnode, label={[yshift=-2.5ex]center:{$q_4(\vec{y}_4)$}}] at (gamma4.south) {};

  \path[-]
    (n1) edge (m2)
    (n2) edge (m3)
    (n2) edge (m4)
  ;
\end{tikzpicture}
\hfill
\begin{tikzpicture}[%
  run/.style = {isosceles triangle, shape border rotate=90, draw=black, isosceles triangle stretches, minimum width=2.5cm, minimum height=1.5cm},
  decoration = {snake, segment length=1mm, amplitude=0.3mm}
]

  \node (alpha)      [run, label={[yshift=-1mm]below:$\alpha$}, anchor=north, fill=blue!50, postaction={pattern=north west lines, pattern color=blue}] {};
  \node (alphasigma) [run, fill=white, label={[yshift=0.25ex]center:$\sigma$}, minimum width=1cm, minimum height=6mm] at (alpha) {};
\end{tikzpicture}
  \caption{%
    Illustration of the quasi-order $\trianglelefteq$ on runs of a BVASS.
    The left part depicts a run $\sigma$ and the middle part depicts a run $\alpha$ such that $\sigma \trianglelefteq \alpha$.
    The two following conditions are met to guarantee that $\sigma \trianglelefteq \alpha$.
    Firstly,
    $\vec{x}_i \leq \vec{y}_i, \vec{z}_i$ for $i = 1, \ldots, 4$.
    This condition ensures that the ordering constraints on the targets holds.
    Secondly,
    $\vec{x}_1 - \vec{x}_2 = \vec{y}_1 - \vec{z}_2$,
    $\vec{x}_2 - (\vec{x}_3 + \vec{x}_4) = \vec{y}_2 - (\vec{z}_3 + \vec{z}_4)$,
    $\vec{x}_3 = \vec{y}_3$ and $\vec{x}_4 = \vec{y}_4$.
    This condition ensures that the equality constraints on the rules holds.
    The right part provides an abstract representation of the middle part.
  }
  \label{fig:wqo-on-runs}
\end{figure}

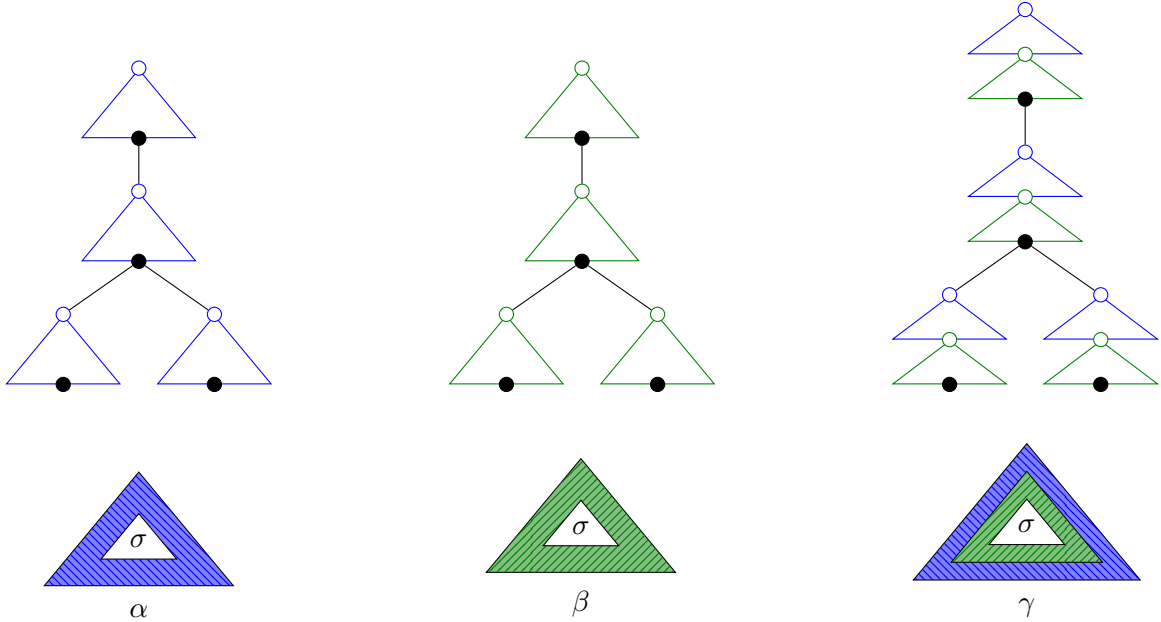
\begin{figure}[t]
  \begin{tikzpicture}[%
  tnode/.style = {circle, draw=black, fill=black, inner sep=0.4ex},
  cnode/.style = {circle, draw=blue, fill=white, inner sep=0.4ex},
  ctx/.style = {isosceles triangle, shape border rotate=90, draw=blue, isosceles triangle stretches, minimum width=1.5cm, minimum height=0.9cm}
]

  \node (gamma1) [ctx] {};
  \node (m1) [cnode] at (gamma1.north) {};
  \node (n1) [tnode] at (gamma1.south) {};

  \node (gamma2) [ctx, anchor=north] at ([yshift=-7mm]n1) {};
  \node (m2) [cnode] at (gamma2.north) {};
  \node (n2) [tnode] at (gamma2.south) {};

  \node (gamma3) [ctx, anchor=north] at ([xshift=-10mm, yshift=-7mm]n2) {};
  \node (m3) [cnode] at (gamma3.north) {};
  \node (n3) [tnode] at (gamma3.south) {};

  \node (gamma4) [ctx, anchor=north] at ([xshift=10mm, yshift=-7mm]n2) {};
  \node (m4) [cnode] at (gamma4.north) {};
  \node (n4) [tnode] at (gamma4.south) {};

  \path[-]
    (n1) edge (m2)
    (n2) edge (m3)
    (n2) edge (m4)
  ;
\end{tikzpicture}
\hfill
\begin{tikzpicture}[%
  tnode/.style = {circle, draw=black, fill=black, inner sep=0.4ex},
  cnode/.style = {circle, draw=green!50!black, fill=white, inner sep=0.4ex},
  ctx/.style = {isosceles triangle, shape border rotate=90, draw=green!50!black, isosceles triangle stretches, minimum width=1.5cm, minimum height=0.9cm}
]

  \node (gamma1) [ctx] {};
  \node (m1) [cnode] at (gamma1.north) {};
  \node (n1) [tnode] at (gamma1.south) {};

  \node (gamma2) [ctx, anchor=north] at ([yshift=-7mm]n1) {};
  \node (m2) [cnode] at (gamma2.north) {};
  \node (n2) [tnode] at (gamma2.south) {};

  \node (gamma3) [ctx, anchor=north] at ([xshift=-10mm, yshift=-7mm]n2) {};
  \node (m3) [cnode] at (gamma3.north) {};
  \node (n3) [tnode] at (gamma3.south) {};

  \node (gamma4) [ctx, anchor=north] at ([xshift=10mm, yshift=-7mm]n2) {};
  \node (m4) [cnode] at (gamma4.north) {};
  \node (n4) [tnode] at (gamma4.south) {};

  \path[-]
    (n1) edge (m2)
    (n2) edge (m3)
    (n2) edge (m4)
  ;
\end{tikzpicture}
\hfill
\begin{tikzpicture}[%
  tnode/.style = {circle, draw=black, fill=black, inner sep=0.4ex},
  cnodecommon/.style = {circle, fill=white, inner sep=0.4ex},
  cnode/.style = {cnodecommon, draw=blue},
  cnodep/.style = {cnodecommon, draw=green!50!black},
  ctxcommon/.style = {isosceles triangle, shape border rotate=90, isosceles triangle stretches, minimum width=1.5cm, minimum height=0cm},
  ctx/.style = {ctxcommon, draw=blue},
  ctxp/.style = {ctxcommon, draw=green!50!black}
]

  \node (gamma1)  [ctx] {};
  \node (gamma1p) [ctxp, anchor=north] at (gamma1.south) {};
  \node (m1)  [cnode] at (gamma1.north) {};
  \node (m1p) [cnodep] at (gamma1p.north) {};
  \node (n1)  [tnode] at (gamma1p.south) {};

  \node (gamma2)  [ctx, anchor=north] at ([yshift=-7mm]n1) {};
  \node (gamma2p) [ctxp, anchor=north] at (gamma2.south) {};
  \node (m2)  [cnode] at (gamma2.north) {};
  \node (m2p) [cnodep] at (gamma2p.north) {};
  \node (n2)  [tnode] at (gamma2p.south) {};

  \node (gamma3)  [ctx, anchor=north] at ([xshift=-10mm, yshift=-7mm]n2) {};
  \node (gamma3p) [ctxp, anchor=north] at (gamma3.south) {};
  \node (m3)  [cnode] at (gamma3.north) {};
  \node (m3p) [cnodep] at (gamma3p.north) {};
  \node (n3)  [tnode] at (gamma3p.south) {};

  \node (gamma4)  [ctx, anchor=north] at ([xshift=10mm, yshift=-7mm]n2) {};
  \node (gamma4p) [ctxp, anchor=north] at (gamma4.south) {};
  \node (m4)  [cnode] at (gamma4.north) {};
  \node (m4p) [cnodep] at (gamma4p.north) {};
  \node (n4)  [tnode] at (gamma4p.south) {};

  \path[-]
    (n1) edge (m2)
    (n2) edge (m3)
    (n2) edge (m4)
  ;
\end{tikzpicture}

\bigskip
\bigskip

\hspace{5mm}%
\begin{tikzpicture}[%
  run/.style = {isosceles triangle, shape border rotate=90, draw=black, isosceles triangle stretches, minimum width=2.5cm, minimum height=1.5cm}
]

  \node (alpha)      [run, label={[yshift=-1mm]below:$\alpha$}, fill=blue!50, postaction={pattern=north west lines, pattern color=blue}] {};
  \node (alphasigma) [run, fill=white, label={[yshift=0.25ex]center:$\sigma$}, minimum width=1cm, minimum height=6mm] at (alpha) {};
\end{tikzpicture}
\hfill
\hspace{2mm}%
\begin{tikzpicture}[%
  run/.style = {isosceles triangle, shape border rotate=90, draw=black, isosceles triangle stretches, minimum width=2.5cm, minimum height=1.5cm}
]

  \node (beta)      [run, label={[yshift=-1mm]below:$\beta$}, fill=green!50!black!50, postaction={pattern=north east lines, pattern color=green!50!black}] {};
  \node (betasigma) [run, fill=white, label={[yshift=0.25ex]center:$\sigma$}, minimum width=1cm, minimum height=6mm] at (beta) {};
\end{tikzpicture}
\hfill
\begin{tikzpicture}[%
  run/.style = {isosceles triangle, shape border rotate=90, draw=black, isosceles triangle stretches, minimum width=2.5cm, minimum height=1.5cm}
]

  \node (gammap)        [run, label={[yshift=-1mm]below:$\gamma$}, fill=blue!50, postaction={pattern=north west lines, pattern color=blue}, minimum width=3cm, minimum height=1.8cm] {};
  \node (gamma)         [run, fill=green!50!black!50, postaction={pattern=north east lines, pattern color=green!50!black}, minimum width=2cm, minimum height=1.2cm] at (gammap) {};
  \node (gammasigma)    [run, fill=white, label={[yshift=0.25ex]center:$\sigma$}, minimum width=1cm, minimum height=6mm] at (gamma) {};

\end{tikzpicture}%
\hspace{1mm}
  \caption{%
    Illustration of the amalgamation property (\cref{lem:ibvas-amalgamation}).
    The left and middle parts respectively depict two runs $\alpha$ and $\beta$ that are above the same run $\sigma$,
    \emph{i.e.},
    such that $\sigma \trianglelefteq \alpha, \beta$.
    The right part depicts the run $\gamma$ resulting from the amalgamation of $\alpha$ and $\beta$.
    The three runs $\alpha$, $\beta$ and $\gamma$ are depicted concretely at the top and abstractly at the bottom.
  }
  \label{fig:amalgamation-of-runs}
\end{figure}

We define the quasi-order $\trianglelefteq$ on runs of a BVASS $\BB$ inductively by $\alpha \trianglelefteq \beta$ if there exists $\beta' \sqsubseteq \beta$ such that
$\tgt{\alpha} \leq \tgt{\beta}, \tgt{\beta'}$,
$\Rule(\alpha) = \Rule(\beta')$ and
$\alpha[i] \trianglelefteq \beta'[i]$ for every $i \in [1, \ra{\alpha}]$. (Note that the condition $\Rule(\alpha) = \Rule(\beta')$ implies that $\ra{\alpha} = \ra{\beta}$.)
It is readily seen that $\trianglelefteq$ is in fact a partial-order (\emph{i.e.}, an antisymmetric quasi-order) on runs of $\BB$.
An illustration of the quasi-order $\trianglelefteq$ is provided in \cref{fig:wqo-on-runs}.

\begin{remark}\label{rem:wqo}
  The quasi-order $\trianglelefteq$ can be defined equivalently via the insertion of contexts which are a branching generalisation of VASS cycles, see~\cite{DBLP:conf/fossacs/BiziereLS26} for such a definition. In \cref{fig:wqo-on-runs}, we present on the left and on the middle two BVASS runs $\sigma$ and $\alpha$ satisfying $\sigma \trianglelefteq \alpha$. Intuitively, $\alpha$ is obtained from $\sigma$ by inserting a family of contexts (depicted as blue triangles) that is hereditarily positive, \emph{i.e.}, such that the total effect of contexts below each node is non-negative (on all components).
\end{remark}

\begin{lemma}
  \label{lem:ibvas-run-order-wpo}
  The pair $(\Runs(\BB), \trianglelefteq)$ is a wqo.
\end{lemma}

\begin{proof}
  To prove \cref{lem:ibvas-run-order-wpo}, we recast $(\Runs(\BB), \trianglelefteq)$ as a tree homeomorphic embedding on decorated runs. Our enriched labels will be taken from the set $\Lambda \coloneqq (Q \times \N^d) \times \Delta \times  (Q \times \N^d)^*$.
  We equip $\Lambda$ with the product order $\leq_{\Lambda} {\coloneqq} \leq_{Q \times \N^d} \times =_{\Delta} \times \leq_{(Q \times \N^d)^*}$, where $\leq_{Q \times \N^d}$ is the product order $=_Q \times \leq_{\N^d}$.
  By Cartesian product and subsequence embedding (see \cref{subsec:wqo}),
  both $(Q \times \N^d, \leq_{Q \times \N^d})$ and $(\Lambda, \leq_{\Lambda})$ are wqos.

  We introduce the (decoration) function $f:\Runs(\BB) \to \TT(\Lambda)$ defined inductively by $f(\alpha)\coloneqq \lambda \langle f(\alpha[1]),\ldots,f(\alpha[k]) \rangle$ with $k \coloneqq \ra{\alpha}$ and $\lambda \coloneqq (\tgt{\alpha}, \Rule(\alpha) , \tgt{\alpha_1}\cdots\tgt{\alpha_k})$.
  By structural induction on $\alpha$, notice that $\alpha \trianglelefteq \beta$ if, and only if, $f(\alpha) \preceq_{\TT(\Lambda)} f(\beta)\wedge \tgt{\alpha}\leq \tgt{\beta}$. It follows that $\trianglelefteq$ is a wqo on $\Runs(\BB)$.
\end{proof}

We will now prove the following \emph{amalgamation} property. For this, we first introduce a substitution operation, which will also help us in the proof of \cref{lem:decompo} in the next section.
The \emph{substitution} in $\alpha$ of the subrun $\alpha' \sqsubseteq \alpha$ by an other run $\beta$, denoted $\alpha[\alpha' \gets \beta]$, is defined only if $\vec{v}\coloneqq \tgt{\alpha'}-\tgt{\beta}\geq \vec{0}$. In this case, the definition is given inductively by $\alpha[\alpha' \gets \beta] = \beta$ if $\alpha = \alpha'$ and else by 
\[q(\vec{x}+\vec{v}) \Bigl \langle \alpha[1],\ldots,\alpha[i-1], \ (\alpha[i]) [\alpha' \gets \beta],  \ \alpha[i+1], \ldots,  \alpha[\ra{\alpha}] \Bigr \rangle\]
where $q(\vec{x})\coloneqq \tgt{\alpha}$, and $i$ is the smallest index such that $\alpha' \sqsubseteq \alpha[i]$.
One can easily show by induction that $\tgt{\alpha[\alpha' \gets \beta]} \ge \tgt{\alpha}$ and $|\alpha[\alpha' \gets \beta]|=|\alpha|-|\alpha'|+|\beta|$.
The above-mentioned amalgamation property is expressed in the following lemma.
An illustration is provided in \cref{fig:amalgamation-of-runs}.

\begin{lemma}
  \label{lem:ibvas-amalgamation}
  For all runs $\sigma, \alpha, \beta$ such that $\sigma \trianglelefteq \alpha, \beta$,
  there exists a run $\gamma$ such that
  $\alpha, \beta \trianglelefteq \gamma$,
  $\lvert \sigma \rvert + \lvert \gamma \rvert = \lvert \alpha \rvert + \lvert \beta \rvert$ and
  $\counters(\sigma) + \counters(\gamma) = \counters(\alpha) + \counters(\beta)$.
\end{lemma}

\begin{proof}

  We prove the lemma by structural induction on $\sigma$.
  Let $\sigma \trianglelefteq \alpha, \beta$ be three runs.
  Following the definition of $\trianglelefteq$, there are subruns $\alpha' \sqsubseteq \alpha$ and $\beta' \sqsubseteq \beta$ such that 
  $\tgt{\sigma} \le \tgt{\alpha}, \tgt{\alpha'}, \tgt{\beta}, \tgt{\beta'}$,
  $\Rule(\sigma) = \Rule(\alpha') = \Rule(\beta')$
  and $\sigma[i] \trianglelefteq \alpha'[i], \beta'[i]$ for every $i \in [1,\ra{\sigma}]$.
  Let $k$ denote the common arity of $\sigma$, $\alpha'$ and $\beta'$ and $\vec a$ be the update in $\Rule(\sigma)$ (\emph{i.e.}, the vector such that $\Rule(\sigma) = (p_1 \cdots p_k, \vec a, q)$ for some $p_1,\ldots,p_k,q \in Q$.)
  By induction hypothesis, for each $i \in [1,k]$ there is a run $\gamma_i$ such that $\alpha'[i], \beta'[i] \trianglelefteq \gamma_i$, 
  $\lvert \sigma[i] \rvert + \lvert \gamma_i \rvert = \lvert \alpha'[i] \rvert + \lvert \beta'[i] \rvert$ 
  and $\counters(\sigma[i]) + \counters(\gamma_i) = \counters(\alpha'[i]) + \counters(\beta'[i])$.
  Summing these equalities over $i=1,\ldots,k$
  and adding $\vec a$ yields
  \begin{equation} \label{eq:sum-amalg}
  \counters(\alpha') + \counters(\beta') = \counters(\sigma) + \vec a + \left (\sum_{i=1}^k \counters(\gamma_i) \right ) 
  \end{equation}

  Let $\gamma' \coloneqq q \left(\vec a + \sum_{i=1}^k \counters(\gamma_i)\right) \left \langle  \gamma_1,\ldots,\gamma_k \right \rangle$ be the run obtained by applying $\Rule(\sigma)$ to the runs $\gamma_i$.
  Using \eqref{eq:sum-amalg}, one can show that the following substitutions are well defined: $\tau \coloneqq \alpha[\alpha' \gets \gamma']$, and $\gamma \coloneqq \beta[\beta' \gets \tau]$.

  There only remains to show that $\alpha, \beta \trianglelefteq \gamma$ since the equality $\lvert \sigma \rvert + \lvert \gamma \rvert = \lvert \alpha \rvert + \lvert \beta \rvert$ is immediate.
  It is readily seen that $\beta' \trianglelefteq \tau$ (by considering the subrun $\gamma' \sqsubseteq \tau$). We derive that $\beta \trianglelefteq \gamma$, using the following observation, which is proved by induction.
  \begin{fact}
    Let $\eta, \theta$ be two runs, and $\eta' \sqsubseteq \eta$ be a subrun. If $\eta' \trianglelefteq \theta$, then $\eta \trianglelefteq \eta[\eta' \gets \theta]$.
  \end{fact}
  The same fact allows us to show that $\alpha \trianglelefteq \tau$. We deduce that $\alpha \trianglelefteq \gamma$ thanks to the following fact (a direct consequence of the definition of $\trianglelefteq$).
  \begin{fact}
    Let $\eta, \theta$ be two runs. If $\tgt{\eta} \le \tgt{\theta}$ and $\eta \trianglelefteq \theta'$ for some $\theta' \sqsubseteq \theta$, then $\eta \trianglelefteq \theta$.
  \end{fact}
  This concludes the proof of the lemma.
\end{proof}

\section{Putting BVASS runs on a diet} \label{sec:rearrangement}

In this section, we establish the main technical tool that will allow us to derive our results in the next sections.
Roughly speaking, we prove that every reachable configuration of a BVASS can be obtained by a run of bounded branching complexity.
The main ingredient in the proof is the amalgamation property established in the previous section (\cref{lem:ibvas-amalgamation}).
Within a run $\alpha$, we use it to replace an immediate subrun $\alpha[i]$ and a subrun $\beta'$ of another immediate subrun $\alpha[j]$, both of which are above a same simple run $\sigma \in \SpS$, by the small $\sigma$ on one side and the larger amalgamated run $\gamma$ on the other side. This transformation preserves the target of $\alpha$ and does not cause any counter to become negative. Repeating this operation progressively makes the run less balanced and more path-like.
An illustration is provided in \cref{fig:rearrangement}.

\begin{figure}[t]
  \begin{tikzpicture}[%
  baseline=(current bounding box.center),
  tnode/.style = {circle, draw=black, fill=black, inner sep=0.3ex},
  run/.style = {isosceles triangle, shape border rotate=90, draw=black, isosceles triangle stretches, minimum width=2.5cm, minimum height=1.5cm},
  decoration = {snake, segment length=1mm, amplitude=0.3mm}
]
  \node (m)  [tnode, label={[yshift=2ex]center:{}}] {};
  \node (ml) [tnode, label={[yshift=2ex, xshift=-3mm]center:{}}] at ([xshift=-14mm, yshift=-5mm]m) {};
  \node (mr) [tnode, label={[yshift=2ex, xshift= 3mm]center:{}}]   at ([xshift= 14mm, yshift=-5mm]m) {};
  \node (n)  [tnode, label={[xshift=-0.75mm]right:{}}] at ([yshift=-5mm]mr) {};

  \path[-]
    (m) edge (ml)
    (m) edge (mr)
    (mr) edge [decorate] (n)
  ;

  \node (alpha)      [run, label={[yshift=-1mm]below:$\alpha[j]$}, anchor=north, fill=blue!50, postaction={pattern=north west lines, pattern color=blue}] at (ml) {};
  \node (beta)       [run, label={[yshift=-1mm]below:$\beta$},  anchor=north] at (mr) {};
  \node (nu)         [run, fill=green!50!black!50, postaction={pattern=north east lines, pattern color=green!50!black}, anchor=north, minimum width=1.67cm, minimum height=1cm] at (n) {};

  \node (alphasigma) [run, fill=white, label={[yshift=0.25ex]center:$\sigma$}, minimum width=1cm, minimum height=6mm] at (alpha) {};
  \node (nusigma)    [run, fill=white, label={[yshift=0.25ex]center:$\sigma$}, minimum width=1cm, minimum height=6mm] at (nu) {};
\end{tikzpicture}
\hfill
\begin{tikzpicture}[baseline=(current bounding box.center), thick]
  \node (start) {};
  \node (end) [right of=start, node distance=8mm] {};

  \path[->] (start) edge (end);
\end{tikzpicture}
\hfill
\begin{tikzpicture}[%
  baseline=(current bounding box.center),
  tnode/.style = {circle, draw=black, fill=black, inner sep=0.3ex},
  run/.style = {isosceles triangle, shape border rotate=90, draw=black, isosceles triangle stretches, minimum width=2.5cm, minimum height=1.5cm},
  decoration = {snake, segment length=1mm, amplitude=0.3mm}
]
  \node (m)  [tnode, label={[yshift=2ex]center:{}}] {}; %
  \node (ml) [tnode, label={[yshift=2ex, xshift=-3mm]center:{}}]  at ([xshift=-14mm, yshift=-5mm]m) {};
  \node (mr) [tnode, label={[yshift=2ex, xshift= 3mm]center:{}}] at ([xshift= 14mm, yshift=-5mm]m) {};
  \node (n)  [tnode, label={[xshift=-0.75mm]right:{}}] at ([yshift=-5mm]mr) {}; %
  \node (p)  at ([yshift=-3mm]n) {};

  \path[-]
    (m) edge (ml)
    (m) edge (mr)
    (mr) edge [decorate] (n)
  ;

  \node (alphasigma) [run, label={[yshift=0.25ex]center:$\sigma$}, anchor=north, minimum width=1cm, minimum height=6mm] at (ml) {};
  \node (betap)      [run, label={[yshift=-1mm]below:}, anchor=north, minimum width=3.33cm, minimum height=2cm] at (mr) {}; %

  \node (nup)        [run, fill=blue!50, postaction={pattern=north west lines, pattern color=blue}, anchor=north, minimum width=2.5cm, minimum height=1.5cm] at (n) {};
  \node (nu)         [run, fill=green!50!black!50, postaction={pattern=north east lines, pattern color=green!50!black}, anchor=north, minimum width=1.67cm, minimum height=1cm] at (p) {};
  \node (nusigma)    [run, fill=white, label={[yshift=0.25ex]center:$\sigma$}, minimum width=1cm, minimum height=6mm] at (nu) {};
\end{tikzpicture}
  \caption{%
    Illustration of the key step in the proof of \cref{lem:decompo}. The run $\alpha$ is depicted on the left, with its $i$th immediate subrun replaced by the equivalent run $\beta$. In our example, $\alpha$ has arity two, $j=1$, $i=2$, and $i,j \in I$.
    To indicate that $\alpha[j], \beta' \trianglerighteq \sigma$ for some subrun $\beta' \sqsubseteq \beta$ and some $\sigma \in \SpS$, we depict $\alpha[j]$ and $\beta'$ by coloured triangles (blue and green respectively) surrounding a smaller triangle labelled $\sigma$ like in \cref{fig:amalgamation-of-runs}. Intuitively, $\alpha[j]$ and $\beta'$ are both equal to $\sigma$ plus some contexts inserted.
    The transformation extracts the contexts from $\alpha[j]$ and incorporates them into $\beta'$, using the amalgamation property \cref{lem:ibvas-amalgamation}.
    The resulting run $\alpha'$, depicted on the right, is less balanced, with a small run $\sigma \in \SpS$ on the left and a bigger run on the right.
  }
  \label{fig:rearrangement}
\end{figure}
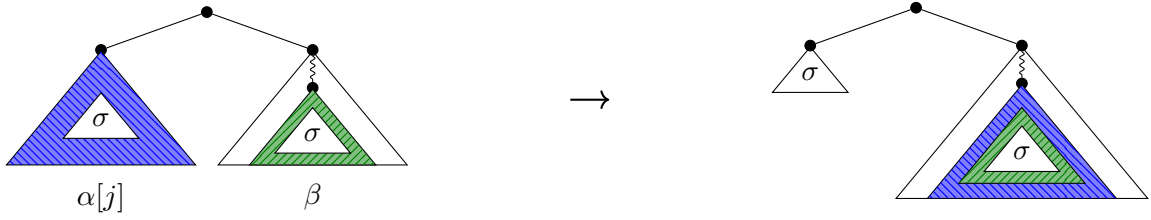

\smallskip

Our bounded branching complexity statement can be formulated more conveniently by introducing an extension of the BVASS under consideration.
Given a BVASS $\BB \coloneqq (Q, \Delta)$ and a finite set $\vec{F} \subseteq \Reach(\BB)$,
we define the \emph{extension of $\BB$ with the configurations in $\vec{F}$} as $\Extend(\BB, \vec{F}) \coloneqq (Q, \Delta')$ where
$\Delta'$ is the (finite) set of triples $(p_{i_1} \cdots p_{i_\ell}, \vec{a}', q)$ in $Q^* \times \Z^d \times Q$ such that there exists a transition rule $(p_1 \cdots p_{k}, \vec a, q)\in \Delta$, $i_1 < \cdots < i_\ell$ in $[1,k]$, and $(\vec x_i)_{i \in [1,k] \backslash \{i_1,\ldots,i_\ell\}}$ such that
    \begin{itemize}
    \item $p_i(\vec x_i) \in \vec F$ for every $i \in [1,k] \backslash \{i_1,\ldots,i_\ell\}$
    \item $\vec a' = \vec a +  \sum_{i \in [1,k] \backslash \{i_1,\ldots,i_\ell\}} \vec x_i$.
    \end{itemize}
Note that $\Delta \subseteq \Delta'$.
Intuitively, $\Extend(\BB, \vec{F})$ has the same runs as $\BB$, except that subruns rooted at configurations in $ \vec{F}$ can be removed.

\begin{lemma} \label{lem:decompo}
  Let $\BB$ be BVASS,
  let $\SpS$ be a finite $\trianglelefteq$-basis of the set of runs of $\BB$, and
  let $\vec{F} \coloneqq \{\target(\sigma) \mid \sigma \in \SpS\}$.
  For every run $\alpha$ of $\BB$, there is a run $\eta$ of $\Extend(\BB, \vec{F})$
  with same target as $\alpha$ and whose branching depth is at most $\card{\SpS} - 1$.
\end{lemma}

\begin{proof}
  Let $\BB$, $\SpS$, and $\vec{F}$ be as in the lemma.
  We define an equivalence relation $\simeq$ on runs of $\BB$ by
  $\alpha \simeq \beta$ if $\target(\alpha) = \target(\beta)$ and $|\alpha| = |\beta|$.
  For every run $\alpha$ of $\BB$,
  we introduce the set
  \[
    \RR(\alpha)
    \coloneqq
    \{\sigma \in \SpS \mid
    \exists \beta \in \Runs(\BB), \alpha \simeq \beta \wedge 
    \exists \beta' \sqsubseteq \beta , \sigma \trianglelefteq \beta'\}
  \]
  We start with three easy observations.
  Firstly,
  $\RR(\alpha) = \RR(\beta)$ for all runs $\alpha, \beta$ of $\BB$ with $\alpha \simeq \beta$.
  This observation directly follows from the definition of $\RR(\alpha)$.
  Secondly,
  $\RR(\alpha) \neq \emptyset$ for every run $\alpha$ of $\BB$.
  This observation follows from the assumption that
  $\SpS$ is a $\trianglelefteq$-basis of the set of runs of $\BB$.
  Thirdly,
  for every run $\alpha$ of $\BB$ and every $\alpha' \sqsubseteq \alpha$,
  we have $\RR(\alpha') \subseteq \RR(\alpha)$.
  Indeed,
  if $\sigma \in \RR(\alpha')$ then
  there exist a run $\beta$ of $\BB$ and a subrun $\beta' \sqsubseteq \beta$
  such that $\alpha' \simeq \beta$ and $\sigma \trianglelefteq \beta'$.
  We may replace in $\alpha$ the subrun $\alpha'$ by $\beta$ as they have the same target.
  The resulting run $\gamma$ has same size and same target as $\alpha$,
  hence,
  $\alpha \simeq \gamma$.
  Moreover,
  $\beta' \sqsubseteq \beta \sqsubseteq \gamma$ and $\sigma \trianglelefteq \beta'$, hence
  $\sigma \in \RR(\alpha)$.

  The lemma is an immediate consequence of the following claim.

  \begin{claim}
    For every run $\alpha$ of $\BB$, there is a run $\eta$ of $\Extend(\BB, \vec{F})$ with same target as $\alpha$ and whose branching depth satisfies $\bd(\eta) < \card{\RR(\alpha)}$.
  \end{claim}

  We prove the claim by induction on $|\alpha|$.
  If $|\alpha| = 1$ then $\bd(\alpha) = 0$,
  hence,
  $\bd(\alpha) < \card{\RR(\alpha)}$ since $\RR(\alpha)$ is not empty
  according to the second observation shown before the claim.
  So the run $\eta := \alpha$ satisfies the claim.
  Now let $K \geq 1$ and assume that the claim holds for every run $\alpha$ of size $|\alpha| \leq K$.
  Pick a run $\alpha$ of size $|\alpha| = K+1$.
  As $|\alpha| \geq 2$, $\alpha$ has arity at least one.
  We distinguish the immediate subruns of $\alpha$ depending on whether or not their target is in $\vec{F}$.
  Formally, let $I \coloneqq \{i \in [1, \ra{\alpha}] \mid \tgt{\alpha[i]} \not\in \vec{F}\}$
  be the set of indices of the immediate subruns of $\alpha$ whose target is not in $\vec{F}$.
  According to the first observation shown before the claim,
  we may assume w.l.o.g. that we picked a run $\alpha$ that minimises $|I|$
  among the runs with same target and same size as $\alpha$.
  We consider two cases depending on $|I|$.
  \smallskip

  The first case is when $|I| \leq 1$.
  So there exists an immediate subrun $\alpha'$ of $\alpha$ such that
  the target of every other immediate subrun of $\alpha$ is in $\vec{F}$.
  We have $\RR(\alpha') \subseteq \RR(\alpha)$
  according to the third observation shown before the claim.
  As $|\alpha'| \leq |\alpha| - 1 = K$,
  we get from the induction hypothesis that
  there is a run $\eta'$ of $\Extend(\BB, \vec{F})$ with
  same target as $\alpha'$ and
  whose branching depth satisfies $\bd(\eta') < \card{\RR(\alpha')}$.
  We transform $\alpha$ into a run $\eta$ of $\Extend(\BB, \vec{F})$ in two steps.
  First, we replace in $\alpha$ the immediate subrun $\alpha'$ by $\eta'$.
  Second, we remove the other immediate subruns of $\alpha$.
  Formally,
  $\eta$ is the tree defined by
  $\eta \coloneqq \tgt{\alpha} \langle \eta' \rangle$.
  It is readily seen that the tree $\eta$ is a run of $\Extend(\BB, \vec{F})$.
  Moreover,
  the arity of $\eta$ is one,
  so $\bd(\eta) = \bd(\eta') < \card{\RR(\alpha')} \leq \card{\RR(\alpha)}$.

  \smallskip

  The second case is when $|I| \geq 2$.
  We have $\RR(\alpha[i]) \subseteq \RR(\alpha)$ for each $i \in I$,
  according to the third observation shown before the claim.
  Let us show that this inclusion is strict.
  Let $i \in I$.
  There exists $j \in I$ with $i \neq j$, as $|I| \geq 2$.
  Since $\SpS$ is a $\trianglelefteq$-basis of $\Runs(\BB)$,
  there exists a run $\sigma \in \SpS$ such that
  $\sigma \trianglelefteq \alpha[j]$.
  Observe that $\sigma \in \RR(\alpha[j]) \subseteq \RR(\alpha)$.
  Let us prove that $\sigma \not\in \RR(\alpha[i])$. The following proof argument is illustrated on \cref{fig:rearrangement}.
  By contradiction,
  suppose that
  there exist a run $\beta$ of $\BB$ and a subrun $\beta' \sqsubseteq \beta$
  such that $\alpha[i] \simeq \beta$ and $\sigma \trianglelefteq \beta'$.
  Let $\gamma$ denote the run obtained by
  applying \cref{lem:ibvas-amalgamation} on $\sigma$, $\alpha[j]$ and $\beta'$.
  We transform $\alpha$ by replacing
  the immediate subrun $\alpha[j]$ by $\sigma$ and
  the immediate subrun $\alpha[i]$ by $\beta[\beta' \gets \gamma]$.
  Let $\alpha'$ denote the resulting tree.
  The properties of $\gamma$ ensure that $\alpha'$ is a run of $\BB$ with same size and same target as $\alpha$.
  However,
  the number of immediate subruns whose target is not in $\vec{F}$ is strictly less in $\alpha'$ than in $\alpha$.
  This contradicts the assumption that our choice of $\alpha$ minimises $|I|$.
  We have shown that
  $\RR(\alpha[i])$ is strictly contained in $\RR(\alpha)$,
  for every $i \in I$.

  Now,
  as $|\alpha[i]| \leq |\alpha| - 1 = K$,
  we get from the induction hypothesis that,
  for each $i \in I$,
  there is a run $\eta_i$ of $\Extend(\BB, \vec{F})$ with same target as $\alpha[i]$ and
  whose branching depth satisfies $\bd(\eta_i) < \card{\RR(\alpha[i])}$.
  It follows that $\bd(\eta_i) \leq \card{\RR(\alpha)} - 2$.
  As before,
  we transform $\alpha$ into a run $\eta$ of $\Extend(\BB, \vec{F})$ in two steps.
  First, we replace in $\alpha$ the immediate subrun $\alpha[i]$ by $\eta_i$, for each $i \in I$.
  Second, we remove every other immediate subrun.
  It is readily seen that the resulting tree is a run of $\Extend(\BB, \vec{F})$.
  The arity of $\eta$ is at least two,
  so $\bd(\eta) = 1 + \max_{i \in I} \bd(\eta_i) < \card{\RR(\alpha)}$.
  This concludes the proof of the claim, and the proof of the lemma.
\end{proof}

A finite $\trianglelefteq$-basis of runs contains all $\trianglelefteq$-minimal runs,
of which there are only finitely many since $\trianglelefteq$ is both a wqo and a partial order.
Unfortunately, minimal runs of a BVASS are not computable. In fact, they are already uncomputable for 2-dimensional VASS \cite{Amalgamation}. Consequently, we cannot even compute the number of minimal runs of a 2-VASS: if this number were computable, we could enumerate all runs, test each one for minimality, and stop once the known number of minimal runs had been found.
It is however still open whether an upper bound on the number of minimal runs of a BVASS is computable. Similarly, it is unknown whether the set of targets of the minimal runs is computable. Consequently, we do not know if \cref{lem:decompo} allows us, given a BVASS $\BB$, to compute a bound $b \in \N$ and a BVASS $\BB'$ with same reachability set as $\BB$ such that every $\vec c \in \Reach(\BB)$ is the target of a run of $\BB'$ with branching depth at most $b$.

It is natural to ask whether the property of $\SpS$ in the hypotheses of \cref{lem:decompo} can be weakened. For instance, would the lemma also hold if $\SpS$ only satisfies the following property
\[ \forall \vec c \in \Reach(\BB), \exists \alpha \in \Runs(\BB), \exists \sigma \in \SpS, \sigma \trianglelefteq \alpha \wedge \tgt{\alpha} = \vec c, \]
\emph{i.e.}, every configuration is reached by a run above some run in $\SpS$?

The following counterexample shows that this hypothesis is not enough.

\begin{figure}
  \centering
  \begin{tikzpicture}[
    >=Stealth,
    node distance=1.5cm,
    every state/.style={minimum size=8mm},
    auto
  ]

  \node[state] (p) {$p$};
  \node[state,right=of p] (q) {$q$};
  \node[left=0.8cm of p] (startp) {};
  \node[right=0.8cm of q] (startq) {};

  \draw[->,red] (p) to[out=125, in=195, looseness=1.2] ++(0,3)
                to[out=15, in=70, looseness=1.2] coordinate[pos=0.7] (pq1) (p);
  \draw[->,red] (q) to[out=90, in=60, looseness=2] coordinate[pos=0.7] (pq2) (p);

  \path[->] (p) edge[bend left]  (q);
  \path[->] (q) edge[bend left]  (p);

  \path[-] (pq2) edge[red] node[red, above] {$+$} (pq1);

  \path[->] (startp) edge node[above] {$0$} (p);
  \path[->] (startq) edge node[above] {$1$} (q);
\end{tikzpicture}
\hspace{2cm}
\begin{tikzpicture}[%
  tnode/.style = {circle, draw=black, fill=black, inner sep=0.4ex}
]

  \node (n1) [tnode, label={left:{$p(0)$}}] {};
  \node (n2) [tnode, label={left:{$q(0)$}}] at ([yshift=-7mm]n1) {};
  \node (n3) [tnode, label={left:{$p(0)$}}] at ([yshift=-7mm]n2) {};

  \node (p1) [tnode, label={left:{$p(1)$}}] at ([yshift=10mm]n1){};
  \node (q1) [tnode, label={right:{$q(1)$}}] at ([yshift=4mm,xshift=10mm]n1){};

  \node (dots) at ([yshift=8mm]p1) {$\vdots$};

  \node (pn) [tnode, label={left:{$p(n)$}}] at ([yshift=5mm]dots){};
  \node (qn) [tnode, label={right:{$q(1)$}}] at ([yshift=-1mm,xshift=10mm]dots){};

  \path[-]
    (n1) edge (n2)
    (n2) edge (n3)
    (p1) edge (q1)
    (pn) edge (qn)
    (p1) edge (n1)
  ;
\end{tikzpicture}
\caption{On the left: a $1$-BVASS. On the right: the run $\alpha_{n}$.}
\label{fig:Scounter}
\end{figure}
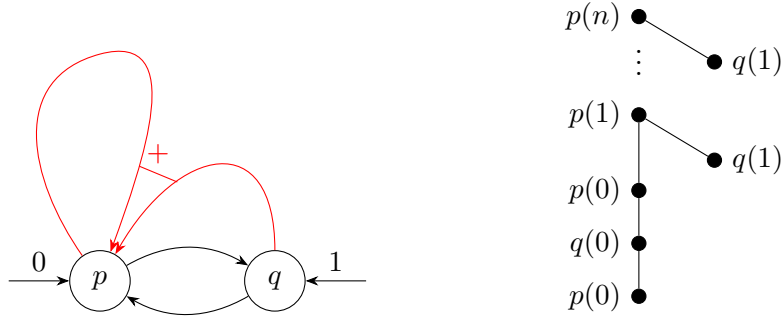

\begin{example}
  Consider the 1-BVASS $\BB$ depicted on the left of \cref{fig:Scounter}. Let $\SpS \coloneqq \{p(0)\langle \rangle, q(0)\langle p(0) \langle \rangle \rangle\}$ and $\vec F \coloneqq \{p(0), q(0)\}$ be the set of targets of runs in $\SpS$. The BVASS $\Extend(\BB, \vec F)$ is the same as $\BB$, but with an additional 0-loop on state $p$. It is easy to see that any run in $\Extend(\BB, \vec F)$ with target $p(n)$ or $q(n)$ must have $n$ leaves labelled $q(1)$. Hence, its branching depth is at least $\lceil\log_2 n\rceil$.

  To check that $\SpS$ satisfies the weak hypothesis stated above, consider the sequences of runs $(\alpha_n)_{n \in \N}$ and $(\beta_n)_{n \in \N}$ defined inductively as depicted in \cref{fig:Scounter} for $\alpha_n$, and by $\beta_n\coloneqq q(n)\langle \alpha_n\rangle$. For all $n \in \N$, we have $p(0) \langle \rangle \trianglelefteq \alpha_n$ and $q(0) \langle p(0) \langle \rangle \rangle \trianglelefteq \beta_n$.

  Thus, we have found a BVASS $\BB$ and a set $\SpS$ of runs satisfying the weak hypothesis such that for every $b \in \N$, there exists $\vec c \in \Reach(\BB)$ which cannot be reached by runs of branching depth less than $b$ in $\Extend(\BB, \vec F)$, where $\vec F$ is the set of targets of runs in $\SpS$.
\end{example}

We end this section with the following corollary. It shows how, from the bound on the branching depth in an equivalent BVASS provided by \cref{lem:decompo}, we can deduce a bound on the Strahler number.
The idea is that the measure ``infimum over all finite $\vec F \subseteq \Reach(\BB)$ of supremum over all $\vec c \in \Reach(\BB)$ of the branching depth needed to reach $\vec c$ in $\Extend(\BB, \vec F)$'' is finer than the measure ``supremum over all $\vec c \in \Reach(\BB)$ of Strahler number needed to reach $\vec c$''.

\begin{corollary} \label{cor:bounded-Strahler}
    For every BVASS $\BB$, there exists $s \in \N$ such that every configuration reached by $\BB$ is reached by a run whose Strahler number is at most~$s$.
\end{corollary}

\begin{proof}

    Let $\BB$ be a BVASS. Let $\SpS$ and $\vec{F}$ be as in \cref{lem:decompo}, and let $b \coloneqq \card{\SpS} - 1$.
    
    For each $\vec c \in \vec F$, let $\sigma_{\vec c}$ be a run with target $\vec c$, and let $s_{\vec c}$ denote its Strahler number.
    We define $s \coloneqq \max_{\vec c \in \vec F}s_{\vec c} + b + 1$ and we show that every configuration reached by $\BB$ is reached by a run of Strahler number at most $s$.

    By \cref{lem:decompo}, every configuration reached by $\BB$ is the target of a run of $\Extend(\BB, \vec F)$ with branching depth at most $b$.
    It is therefore sufficient to define a target-preserving function mapping runs $\alpha$ of $\Extend(\BB, \vec F)$ with branching depth at most $b$ to runs of $\BB$ with Strahler number at most $s$.
    This function is very natural: we simply re-attach to $\alpha$ the runs $\sigma_{\vec c}$ that the additional rules of $\Extend(\BB, \vec F)$ allowed to cut.

    Formally, the function $\texttt{reattach}$ is defined by structural induction as follows.
    Let $\alpha \coloneqq \vec c \langle \alpha_1,\ldots,\alpha_k \rangle$ be a run of $\Extend(\BB, \vec F)$.
    Let $(p_1 \cdots p_{k'}, \vec a, q)$ be a rule of $\BB$, $i_1 < \cdots < i_k$ be elements of $[1,k']$, and $(\vec x_i)_{i \in [1,k'] \backslash \{i_1,\ldots,i_k\}}$ be such that
    \begin{itemize}
    \item $p_i(\vec x_i) \in \vec F$ for every $i \in [1,k'] \backslash \{i_1,\ldots,i_k\}$
    \item $\Rule(\alpha) = (p_{i_1} \cdots p_{i_k}, \vec a + \sum_{i \in [1,k'] \backslash \{i_1,\ldots,i_k\}} \vec x_i, q)$.
    \end{itemize}
    We define $\mathtt{reattach}(\alpha) \coloneqq \vec c \langle \beta_1, \ldots ,\beta_{k'} \rangle$ where $\beta_i$ is $\mathtt{reattach}(\alpha_i)$ if $i \in \{i_1,\ldots,i_k\}$ and otherwise $\sigma_{p_i(\vec x_i)}$.

    The rule, the indices $i_1,\ldots,i_k$, and the vectors $\vec x_i$ in the above paragraph exist by definition of $\Extend(\BB,\vec F)$, although they need not be unique. Different choices may therefore result in different runs $\mathtt{reattach}(\alpha)$.
    It is immediate to see by structural induction that $\mathtt{reattach}(\alpha)$ is a run of $\BB$ with same target as $\alpha$.
    We now show by structural induction that for every run $\alpha$ of $\Extend(\BB, \vec F)$, $\sn(\mathtt{reattach}(\alpha)) \le \max_{\vec c \in \vec F} s_{\vec c} + \bd(\alpha) + 1$.

    If $\alpha$ only has one immediate subrun $\alpha'$, then by induction hypothesis $\sn(\mathtt{reattach}(\alpha')) \le \max_{\vec c \in \vec F} s_{\vec c} + \bd(\alpha') + 1 \le \max_{\vec c \in \vec F} s_{\vec c} + \bd(\alpha) + 1$. The immediate subruns of $\mathtt{reattach}(\alpha)$ consist of $\mathtt{reattach}(\alpha')$ and possibly some runs $\sigma_{\vec c}$, whose Strahler number is at most $\max_{\vec c \in \vec F} s_{\vec c} < \max_{\vec c \in \vec F} s_{\vec c} + \bd(\alpha) + 1$. Thus, $\sn(\mathtt{reattach}(\alpha)) \le \max_{\vec c \in \vec F} s_{\vec c} + \bd(\alpha) + 1$.

    If $\alpha$ has at least two immediate subruns, then their branching depth is at most $\bd(\alpha)-1$, and by induction hypothesis the Strahler number of their ``re-attached'' version is at most $\max_{\vec c \in \vec F} s_{\vec c} + \bd(\alpha)$. The other immediate subruns of $\mathtt{reattach}(\alpha)$ also have Strahler number at most $\max_{\vec c \in \vec F} s_{\vec c} + \bd(\alpha)$ (and even $\max_{\vec c \in \vec F} s_{\vec c}$, since they are $\sigma_c$'s). Hence, $\sn(\mathtt{reattach}(\alpha)) \le \max_{\vec c \in \vec F} s_{\vec c} + \bd(\alpha) + 1$, completing the induction.
\end{proof}

\section{BVASS reachability sets are VASS sections} \label{sec:section-VASS}

In this section, we use the technical results of \cref{sec:rearrangement} to prove our main theorem, \cref{thm:BVASS-reachsets-are-VASS-sections}, and we then derive two corollaries.
We begin with a lemma stating that the set of configurations reachable by runs of bounded branching depth in a BVASS is a VASS section.
For every BVASS $\BB$ and natural number $b$, denote by $\Reach_{\bd \leq b}(\BB)$ the targets of runs of $\BB$ whose branching depth is at most $b$.

\begin{lemma} \label{lem:construction-VASS}
  For every $d$-BVASS $\BB$ and $b \in \N$, one can compute a VASS $\VV_b$ such that $\Sect_{[1,d]}^Q(\VV_b)=\Reach_{\bd \leq b}(\BB)$.
\end{lemma}

\begin{proof}
  We first observe that VASS sections are effectively closed under finite unions. In fact, we can define an algorithm that takes as input two VASS $\VV$ and $\VV'$ and returns a VASS denoted by $\VV\sqcup\VV'$ such that $\Sect_{[1,d]}^Q(\VV\sqcup\VV')=\Sect_{[1,d]}^Q(\VV)\cup \Sect_{[1,d]}^Q(\VV')$. The VASS $\VV\sqcup\VV'$ is obtained as follows. First we consider the union of $\VV$,  $\VV'$ and the set of states $Q$. In this union the states of $\VV$ and $\VV'$ are renamed in such a way states are pairwise disjoint. By adding to this VASS zero-effect transitions from the two copies of each state $q\in Q$ to $q$, we get $\VV\sqcup\VV'$ satisfying the required property.

\smallskip

Next, we provide an algorithm that produces a VASS capturing the effect of a unique transition rule $\delta=(p_1 \cdots p_k,\vec{a},q)$. Given a set $\vec{C}$ of configurations, we introduce the set $\delta[\vec{C}]$ of configurations $q(\vec{x})$ such that there exist $p_1(\vec x_1),\ldots,p_k(\vec x_k) \in \vec C$ satisfying $\vec x = \vec a + \sum_{i=1}^k \vec x_i$. Starting from a VASS $\VV$ such that $\vec{C}=\Sect_{[1,d]}^Q(\VV)$, we construct a VASS $\delta[\VV]$ such that $\delta[\vec{C}]=\Sect_{[1,d]}^Q(\delta[\VV])$, as follows. Consider the Cartesian product of $k$ copies of $\VV$ with sets of counters pairwise disjoint and disjoint from $[1,d]$. As expected, the states of this Cartesian product are tuples $(q_1,\ldots,q_k)\in Q^k$. We add to this product the state $q$ (the target of $\delta$) and a fresh state $q_!$ serving as an intermediate state between the Cartesian product and $q$. We also add zero-effect transitions from the $(p_1,\ldots,p_k)\in Q^k$ to $q_!$, and a transition from $q_!$ to $q$ that adds $\vec{a}$ to the counters $[1,d]$. In order to transfer the counters from each copy of $\VV$ to $[1,d]$, we add $dk$ loops on $q_!$ that decrement a counter of some copy of $\VV$ while increasing the corresponding counter in $[1,d]$. It is routinely checked that the VASS $\delta[\VV]$ defined in this way satisfies $\delta[\vec{C}]=\Sect_{[1,d]}^Q(\delta[\VV])$.

\smallskip

We can now construct $\VV_b$ inductively on $b$. The VASS $\VV_0$ is simply the VASS obtained from $\BB$ by removing all branching transition rules. To construct $\VV_{b+1}$, we first let $\VV_b'=\sqcup_{\delta\in\Delta}\delta[\VV_b]$. We then add the unary transition rules of $\BB$ to $\VV_b'$ to obtain $\VV_{b+1}$.
\end{proof}

We can now state and prove the main theorem.

\begin{theorem} \label{thm:BVASS-reachsets-are-VASS-sections}
    Every reachability set of a BVASS is a section of VASS.
\end{theorem}
\begin{proof}
  Consider a $d$-BVASS $\BB$.
  The set of runs of $\BB$ admits a finite $\trianglelefteq$-basis $\SpS$ since
  $\trianglelefteq$ is a wqo.
  Let $b \coloneqq \card{\SpS} - 1$,
  let $\vec{F} \coloneqq \{\target(\sigma) \mid \sigma \in \SpS\}$, and
  let $\EE \coloneqq \Extend(\BB, \vec{F})$.
  By \cref{lem:decompo},
  we have $\Reach(\BB) \subseteq \Reach_{\bd \leq b}(\EE)$.
  It follows that $\Reach(\BB) = \Reach_{\bd \leq b}(\EE)$ since
  $\Reach_{\bd \leq b}(\EE) \subseteq \Reach(\EE) = \Reach(\BB)$ by definition.
  By applying \cref{lem:construction-VASS} to $\EE$,
  we conclude that $\Reach(\BB) = \Reach_{\bd \leq b}(\EE)$ is a section of VASS.
\end{proof}

With the same techniques as for \cref{thm:BVASS-reachsets-are-VASS-sections}, we can also prove the following proposition.

\begin{proposition}\label{prop:semilinearity}
   If the reachability set of a BVASS is semilinear, then it can be effectively computed.
\end{proposition}

\begin{proof}
  Let $\BB$ be a $d$-BVASS whose reachability set is semilinear, and let $(\vec c_n)_{n \in \N}$ be an enumeration of its reachability set.
  By \cref{lem:construction-VASS}, for each $n \in \N$, we can compute a VASS $\VV_n$ such that $\Sect_{[1,d]}^Q(\VV_n) = \Reach_{\bd \le n}(\Extend(\BB, \{c_0,\ldots,c_n\}))$.
  We use Corollary VII.7 and Theorem VII.11 of~\cite{CGL25}, which together imply that it is decidable whether a VASS section is semilinear and, if it is, that a semilinear presentation can be computed.
  
  The reachability set of $\BB$ can be computed by repeating the following instruction for $n \in \N$.
  \begin{itemize}
    \item Compute VASS $\VV_n$.
    \item For each $\VV_n$, test whether $\Sect_{[1,d]}^Q(\VV_n)$ is semilinear.
    \item If it is semilinear, compute a semilinear presentation.
    \item Test whether this semilinear set is an inductive invariant for $\BB$, and if yes return it.
  \end{itemize}

  To see that this algorithm is correct, observe first that \[\Sect_{[1,d]}^Q(\VV_n) = \Reach_{\bd \le n}(\Extend(\BB, \{c_0,\ldots,c_n\})) \subseteq \Reach(\BB).\]
  Hence, $\Sect_{[1,d]}^Q(\VV_n) = \Reach(\BB)$ if and only if $\Sect_{[1,d]}^Q(\VV_n)$ is semilinear and is an inductive invariant of $\BB$. 
  Moreover, by \cref{lem:decompo}, there exists $n \in \N$ such that
   $\Reach(\BB)$ is equal to $\Reach_{\bd \le n}(\Extend(\BB, \{c_0,\ldots,c_n\}))$ , so the algorithm terminates.
\end{proof}

\begin{remark}
  The equality of two VASS reachability sets is undecidable~\cite{HACK197677}. Consequently, although the sequence of VASS $(\VV_n)_{n\in\N}$ constructed in the previous proof can be effectively generated for every BVASS, we cannot in general detect whether $\Sect{[1,d]}^Q(\VV_n)$ has reached the full reachability set of $\BB$. Hence, this sequence does not directly yield a decision procedure for BVASS reachability.
\end{remark}

The reachability set of a BVASS is not semilinear in general (this is already the case for VASS). We nevertheless show that it is almost semilinear. This geometric property was established for VASS and used in~\cite{Turing-100:Vector_Addition_Systems_Reachability} to show that, for every VASS and every unreachable configuration, there exists a semilinear inductive invariant excluding that configuration.

\smallskip

Formally, an \emph{almost semilinear set} of $Q\times\N^d$ is a finite union of sets of the form $q(\vec{b}+\vec{P})$ where $q\in Q$ is a state, $\vec{b}$ is a vector in $\N^d$ called the \emph{basis}, and $\vec{P}$ is a subset of $\N^d$ called the \emph{periodic set} and satisfying $\vec{0}\in\vec{P}$, $\vec{P}+\vec{P}\subseteq \vec{P}$ and such that the \emph{cone} $\Q_{\geq 0}\vec{P}=\{\lambda\vec{p} \mid \lambda\in\Q_{\geq 0}\wedge \vec{p}\in\vec{P}\}$ spanned by $\vec{P}$ is definable in $FO(\Q_{\geq 0},+,=)$. In~\cite{Turing-100:Vector_Addition_Systems_Reachability}, intersections of VASS reachability sets with semilinear sets were proved to be almost semilinear. Since VASS sections are intersections of VASS reachability sets with semilinear sets and semilinear sets are stable by finite intersections, we deduce that intersections of VASS sections with semilinear sets are also almost semilinear. From the previous \cref{thm:BVASS-reachsets-are-VASS-sections}, we deduce the following corollary.
\begin{corollary}\label{cor:BVASS-reachsets-are-almostsemilinear}
  Intersections of BVASS reachability sets with semilinear sets are almost semilinear.
\end{corollary}

\section{Strahler-Bounded Reachability} \label{sec:vassnz}
We prove that the following bounded-variant of the reachability problem for BVASS is decidable.
\decisionproblem{Strahler-Bounded Reachability}{A BVASS $\BB$, a configuration $q(\vec{x})$, and $s\in\N$}{Is the configuration $q(\vec{x})$ reachable by a run of Strahler number at most $s$ ?}
This problem is motivated by \cref{cor:bounded-Strahler} since for every BVASS $\BB$ there exists a minimal $s\in\N$, called the \emph{Strahler number} of $\BB$, such that every reachable configuration is reachable by a run of Strahler number at most $s$.

\medskip

Given a $d$-BVASS $\BB \coloneqq (Q, \Delta)$%
, we define a VASS with nested zero tests $\ZZ_s$ such that $\Sect_{[1,d]}^Q(\ZZ_s)$ coincides with the set of configurations of $\BB$ reachable by a run whose Strahler number is at most $s$. 
\emph{VASS with nested zero tests} (abbreviated \emph{VASSnz}) extend classical VASS by allowing restricted forms of zero tests: transitions may simultaneously test for zero all counters belonging to some counter subset $I_k$ in an increasing sequence $I_1 \subset I_2 \subset \dots \subset I_K$. 
Membership in a VASSnz section is decidable~\cite[corollary VII.7, theorem VII.11]{CGL25}, hence this construction is enough to prove decidability of the above problem.

\smallskip

Instead of describing $\ZZ_s$ by explicitly listing its states and transitions, we specify it by means of \cref{algo:reach_vassnz,algo:reach_vassnz0}. These algorithms can easily be compiled into a VASSnz, as they only involve state comparisons, constant additive counter updates, zero tests, finite branching, nondeterministic choices, and function calls with bounded nesting depth.
The $\star$ symbol at lines $4$ and $10$ denotes a non-deterministic choice condition (this condition is non-deterministically evaluated to false or true at each iteration of the \textbf{while}-loop). The \textbf{abort} instruction is non terminating (a state of the VASSnz without any outgoing transition).

\SetKw{Abort}{abort}
\SetKw{Assert}{assert}
\SetKw{Call}{call}
\SetKw{Pick}{pick}

\newcommand{\add}[2]{{#1} \ {+}{=} \ {#2}}

\begin{algorithm}[t]
  \DontPrintSemicolon
  \Pick $(\varepsilon, \vec a, q)$ in $\Delta$ with $\vec a \geq \vec 0$\; \label{l:initS}
  $\mathtt{state}^{(s)} := q$\;
  $\add{(\mathtt{c}_1^{(s)}, \ldots, \mathtt{c}_d^{(s)})}{(\vec{a}(1), \ldots, \vec{a}(d))}$\; \label{l:initE}
  \While{$\star$ \label{l:transS}}{
    \Pick $(p_1 \cdots p_k, \vec a, q)$ in $\Delta$ and $j_0$ in $[1,k]$ such that $\mathtt{state}^{(s)} = p_{j_0}$\; 
    \For{$j \in [1,k]$}{
      \If{$j \neq j_0$}{
      \Call{$\mathtt{REACH}_{p_{j}}^{(s-1)}$}\;
      \While{$\star$ \label{l:transferS}}{
        \Pick $i$ in $[1, d]$\;
        $\add{(\mathtt{c}_i^{(s-1)}, \mathtt{c}_i^{(s)})}{(-1, 1)}$\; \label{l:transferE}
      }
      \Assert{$\mathtt{c}_1^{(1)} = \cdots = \mathtt{c}_d^{(1)} = \cdots = \mathtt{c}_1^{(s-1)} = \cdots = \mathtt{c}_d^{(s-1)} = 0$}\; \label{l:zerotest}
    }}
    $\mathtt{state}^{(s)} := q$\;
    $\add{(\mathtt{c}_1^{(s)}, \ldots, \mathtt{c}_d^{(s)})}{(\vec{a}(1), \ldots, \vec{a}(d))}$\; \label{l:transE}
  }
  \If{$\mathtt{state}^{(s)} = t$ \label{l:retS}}{
    \Return\; \label{l:return} %
  }
  \Else{
    \Abort\; \label{l:retE}
  }

  \caption{$\mathtt{REACH}_t^{(s)}$ for $t \in Q$ and $s \in \N$}
  \label{algo:reach_vassnz}
\end{algorithm}

\begin{algorithm}[t]
  \DontPrintSemicolon
  \Abort\;
  \caption{$\mathtt{REACH}_t^{(-1)}$ for $t \in Q$}
  \label{algo:reach_vassnz0}
\end{algorithm}

\smallskip

The algorithm $\mathtt{REACH}_t^{(s)}$ simulates a bottom-up traversal of a run of $\BB$ whose target state is $t$ and whose Strahler number is at most $s$ \footnote{The reason why there is a different algorithm $\mathtt{REACH}_t^{(s)}$ for every target state $t$ is only to simplify the recursive calls.}. It first nondeterministically selects an initial configuration (lines \ref{l:initS}–\ref{l:initE}), then performs an arbitrary number of transition rules of $\BB$ (lines \ref{l:transS}–\ref{l:transE}), and finally reaches the return line if the current control state is $t$ (lines \ref{l:retS}–\ref{l:retE}).

\smallskip

To simulate a branching transition rule of $\BB$, $\mathtt{REACH}_t^{(s)}$ makes the required number of calls to $\mathtt{REACH}_p^{(s-1)}$ (namely, the arity of the branching rule minus one) on auxiliary counters. It then transfers the content of these auxiliary counters back to the main ones (line \ref{l:transferS}–\ref{l:transferE}), and finally uses a zero test (line \ref{l:zerotest}) to ensure that all auxiliary counters have been completely emptied. Zero tests play a crucial role here: since auxiliary counters may be reused for several calls to $\mathtt{REACH}_p^{(s-1)}$, any residual value left in them could incorrectly enable transitions during subsequent calls.

\smallskip

We denote by $\ZZ_s$ the VASS obtained by compiling the algorithms $\mathtt{REACH}_t^{(s)}$ for all $t \in Q$ and taking their disjoint union.
More precisely, observe that each algorithm $\mathtt{REACH}_t^{(s)}$ can be compiled into a VASS with nested zero tests $\ZZ_{s,t}$ having $(s+1)d$ counters, one for each integer variable $\mathtt{c}_i^{(j)}$, $i \in [1,d]$, $j \in [0,s]$. We index them so that the variables $(\mathtt{c}_1^{(s)}, \ldots, \mathtt{c}_d^{(s)})$ correspond to indices $1,\ldots,d$.
The states of $\ZZ_{s,t}$ encode the current values of each variable $\mathtt{state}^{(j)}$, the height of the bounded call stack —\emph{i.e.}, which $\mathtt{REACH}_p^{(j)}$ we are currently executing—, and the current line in this $\mathtt{REACH}_p^{(j)}$.
The VASS $\ZZ_s$ has the same counters as the $\ZZ_{s,t}$. Its set of states is a disjoint union of $Q$ with the set of states of the $\ZZ_{s,t}$. It has all the transitions of the $\ZZ_{s,t}$, as well as additional transitions of effect zero allowing to jump to state $t$ from any state of $\ZZ_{s,t}$ corresponding to the return line of $\mathtt{REACH}_t^{(s)}$.

The VASS $\ZZ_s$ satisfy the following fact.

\begin{fact} \label{fact:construction-Zs}
  For all $s \in \N$ and $t \in Q$, the set of counter values $\vec x$ such that $t(\vec x) \in \Sect_{[1,d]}^Q(\ZZ_s)$ coincides exactly with the set of values that $(\mathtt{c}_1^{(s)},\ldots,\mathtt{c}_d^{(s)})$ may hold when executing the return line (line \ref{l:return}) in $\mathtt{REACH}_t^{(s)}$.
\end{fact}

\begin{lemma} \label{lem:construction-VASSnz}
  $\Sect_{[1,d]}^Q(\ZZ_s)$ is the set of configuration reachable by a run of $\BB$ of Strahler number at most $s$.
\end{lemma}

 \begin{proof}
  For each $s \in \N$, let $\Reach_s(\BB)$ denote the set of configurations reachable by a run of $\BB$ whose Strahler number is at most~$s$.
  Recall that by construction $\ZZ_s$ satisfy \cref{fact:construction-Zs}.
  We will also use the following fact, which can easily be proved by induction on $\vec s$.

  \begin{fact} \label{fact:Reach-monotone}
    For all natural numbers $s < s'$, all $\vec x \in \N^d$ and all states $t$, if there is an execution of $\mathtt{REACH}_t^{(s)}$ in which the variables $(\mathtt{c}_1^{(s)},\ldots, \mathtt{c}_d^{(s)})$ hold the value $\vec x$ when executing the return line, then there is also an execution of $\mathtt{REACH}_t^{(s')}$ in which the variables $(\mathtt{c}_1^{(s')},\ldots, \mathtt{c}_d^{(s')})$ hold the value $\vec x$ when executing the return line.
  \end{fact}

  We start by proving the inclusion $\Reach_s(\BB) \subseteq \Sect_{[1,d]}^Q(\ZZ_s)$ for all $s \in \N$.
  Thanks to \cref{fact:construction-Zs}, it is enough to prove by structural induction that for every run $\alpha$ with target $t(\vec x)$, there is an execution of $\mathtt{REACH}_t^{(\sn(\alpha))}$ that reaches the return line with the the variables $(\mathtt{c}_1^{(\sn(\alpha))},\ldots, \mathtt{c}_d^{(\sn(\alpha))})$ holding the value $\vec x$.
  So let $\alpha$ be a run of $\BB$. If $\ra{\alpha} = 0$, then $(\varepsilon, \vec a, q)$ is a rule of $\BB$ and it is easy to see that there is an execution of $\mathtt{REACH}_q^{(0)}$ which never enters the while loop and reaches the return line with the the variables $(\mathtt{c}_1^{0},\ldots, \mathtt{c}_d^{0})$ holding the value $\vec a$.
  We can therefore assume that $\alpha$ has at least one child.
  
  Observe that $\alpha$ has at most one child with same Strahler number.
  If $\alpha$ has a child $\alpha'$ with same Strahler number, then let $p(\vec y)$ denote its target. By induction hypothesis, there is an execution of $\mathtt{REACH}_p^{(\sn(\alpha))}$ that reaches the return line with the the variables $(\mathtt{c}_1^{(\sn(\alpha))},\ldots, \mathtt{c}_d^{(\sn(\alpha))})$ holding the value $\vec y$.
  If $\alpha$ has no child with same Strahler number, then let $\alpha'$ be any of its children, and denote again its target $p(\vec y)$. By induction hypothesis, there is a execution of $\mathtt{REACH}_p^{(\sn(\alpha'))}$ that reaches the return line with the the variables $(\mathtt{c}_1^{(\sn(\alpha'))},\ldots, \mathtt{c}_d^{(\sn(\alpha'))})$ holding the value $\vec y$. Since $\sn(\alpha') < \sn(\alpha)$, we can apply \cref{fact:Reach-monotone}. Hence, in both cases, we have an execution of $\mathtt{REACH}_p^{(\sn(\alpha))}$ that reaches the return line with the the variables $(\mathtt{c}_1^{(\sn(\alpha))},\ldots, \mathtt{c}_d^{(\sn(\alpha))})$ holding the value $\vec y$.

  We now describe how to turn this execution into an execution of $\mathtt{REACH}_t^{(\sn(\alpha))}$ that reaches the return line with the variables $(\mathtt{c}_1^{(\sn(\alpha))},\ldots, \mathtt{c}_d^{(\sn(\alpha))})$ holding the value $\vec x$.
  First, take the prefix of this execution until the last time when the condition of the while loop was evaluated (line \ref{l:transS}). Change the result of this evaluation to true. (Remember that this condition is simply a nondeterministic choice between true or false.) At the next line, pick $\Rule(\alpha)$ and pick $j_0$ such that $\alpha'$ is the $j_0^{\text{th}}$ child of $\alpha$.
  By induction hypothesis combined with \cref{fact:Reach-monotone}, for every other child $\alpha''$ with target $q(\vec z)$, there is an execution of $\mathtt{REACH}_q^{(\sn(\alpha)-1)}$ that reaches the return line with the the variables $(\mathtt{c}_1^{(\sn(\alpha) -1)},\ldots,\mathtt{c}_d^{(\sn(\alpha) -1)} )$ holding the value $\vec z$. We can therefore prolong the execution such that it exits the while loop and reaches line \ref{l:retS} with the variables $(\mathtt{c}_i^{(\sn(\alpha))}, i \in [1,d])$ holding the value $\vec x$. Now, observe that even though the execution we started with is an execution of $\mathtt{REACH}_p^{(\sn(\alpha))}$ instead of $\mathtt{REACH}_t^{(\sn(\alpha))}$, this difference in the state does not matter before reaching line \ref{l:retS}. So let us view our execution as an execution of $\mathtt{REACH}_t^{(\sn(\alpha))}$. Hence, the if-condition at \ref{l:retS} evaluates to true and the return is reached, which finishes the proof of the inclusion $\Reach_s(\BB) \subseteq \Sect_{[1,d]}^Q(\ZZ_s)$ for all $s \in \N$.

  \smallskip

  For the converse inclusion, one can show by induction \footnote{The elements of $\N^2$ are ordered lexicographically.} on $(s, \ell) \in \N^2$ that for every state $t \in Q$, in every execution of $\mathtt{REACH}_t^{(s)}$, after the $\ell^{\text{th}}$ execution of the while loop (if there are at least $\ell$ iterations), the variables $\mathtt{state}^{(s)}$ and  $(\mathtt{c}_1^{(s)},\ldots, \mathtt{c}_d^{(s)})$ hold values $q$ and $\vec x$ such that $q(\vec x) \in \Reach_s(\BB)$.

 \end{proof}

Observe that the existence of a computable upper bound on the Strahler number of a BVASS would imply decidability of the reachability problem for BVASS.
It is still open whether such a bound exists. However, an easy reduction from inclusion of VASS reachability sets already shows that exact computation of the Strahler number is impossible.

\begin{proposition}\label{prop:Uncomputability-Strahler}
    For every $s\in \N$, it is undecidable whether, given a BVASS $\BB$, the Strahler number of $\BB$ is inferior or equal to $s$.
\end{proposition}
\begin{proof}
  We reduce to the inclusion problem for VASS reachability set. This problem is known to be undecidable~\cite{HACK197677,JANCAR-inclusion-VASS-reachset}. Let $\VV_1 \coloneqq (Q, \Delta_1,\vec S_1)$ and $\VV_2 := (Q, \Delta_2,\vec S_2)$ be two VASS of a same dimension $d$ and on a same set of states~$Q$. We define a BVASS $\BB_s \coloneqq (Q_s, \Delta_s)$ which has Strahler number at most $s$ if and only if $\Reach(\VV_1) \subseteq \Reach(\VV_2)$.

    The state set of $\BB_s$ is the product of $Q$ with the nodes of a complete binary tree of depth~$s+1$, that is, $Q_s \coloneqq Q \times Q_{\text{tree}}$ with %
    \[Q_{\text{tree}} := \{ w \mid w \in \{0,1\}^{\le s+1}\}.\]

    (We denote the set of words of length at most $\ell$ on alphabet $\{0,1\}$ by $\{0,1\}^{\le \ell}$ and the empty word by $\varepsilon$.)
    
    Intuitively, we place $\VV_2$ at the root $\varepsilon$, and $\VV_1$ at the leftmost one least $0^{s+1}$.
    All other leaves are constrained to have counters equal to~$0$.
    The branching transition rules of $\BB_s$ mirror the structure of the tree and serve to propagate configurations from $\VV_1$ to $\VV_2$, which requires trees of Strahler number $s+1$. Formally we define:
    \[ \Delta_s \coloneqq \Delta_{\text{init},0} \cup \Delta_{\text{init},1} \cup \Delta_{\text{init},2} \cup \Delta_{\text{tree}} \cup \Delta_1' \cup \Delta_2' \]
    with
    \[
\begin{aligned}
\Delta_{\text{init},0} &\coloneqq 
  \bigl\{\, (\varepsilon, \vec{0}, (q,w)) 
  \;\big|\; q \in Q,\; w \in \{0,1\}^{s+1} \setminus \{0^{s+1}\} \,\bigr\}, \\[4pt]
\Delta_{\text{init},1} &\coloneqq 
  \bigl\{\, (\varepsilon, \vec{s}_1, (q, 0^{s+1})) 
  \;\big|\; q(\vec{s}_1) \in \vec{S}_1 \,\bigr\}, \\[4pt]
\Delta_{\text{init},2} &\coloneqq 
  \bigl\{\, (\varepsilon, \vec{s}_2, (q, \varepsilon)) 
  \;\big|\; q(\vec{s}_2) \in \vec{S}_2 \,\bigr\}, \\[4pt]
\Delta_{\text{tree}} &\coloneqq 
  \Bigl\{\,\Bigl ( (q,w0) (q,w1), \vec{0}, (q,w) \Bigr ) 
  \;\big|\; q \in Q,\; w \in \{0,1\}^s \,\Bigr\}, \\[4pt]
\Delta_1' &\coloneqq 
  \bigl\{\, ((p, 0^{s+1}), \vec{a}_1, (q,0^{s+1})) 
  \;\big|\; (p,\vec{a}_1,q) \in \Delta_1 \,\bigr\}, \\[4pt]
\Delta_2' &\coloneqq 
  \bigl\{\, ((p, \varepsilon), \vec{a}_2, (q,\varepsilon)) 
  \;\big|\; (p,\vec{a}_2,q) \in \Delta_2 \,\bigr\}.
\end{aligned}
\]

$\VV_1$ is simulated by the set of initial configurations $\Delta_{\text{init},1}$ and the set of unary transition rules $\Delta_1'$. Similarly, $\VV_2$ is simulated by $\Delta_{\text{init},2}$ and $\Delta_2'$. The tree structure consists of initial configurations for other leaves ($\Delta_{\text{init},0}$) and by the branching transition rules $\Delta_{\text{tree}}'$.

We now verify that
$\BB_s$ has Strahler number at most $s$ if and only if $\Reach(\VV_1) \subseteq \Reach(\VV_2)$.

If $\Reach(\VV_1) \not \subseteq \Reach(\VV_2)$, then let $q(\vec x) \in \Reach(\VV_1) \backslash \Reach(\VV_2)$. Then in $\BB_s$ the configuration $(q, \varepsilon) (\vec x)$ can only be reached by runs of Strahler number $s+1$. (Intuitively, we need to produce $q(\vec x)$ in the simulation of $\VV_1$ and propagate it to the root. The propagation requires Strahler number $s+1$.)

If $\Reach(\VV_1) \subseteq \Reach(\VV_2)$, then the set of configurations $q(\vec x)$ of $\VV_2$ such that $(q, \varepsilon) (\vec x) \in \Reach(\BB_s)$ is equal to $\Reach(\VV_2)$. Therefore, every configuration in $\Reach(\BB_s)$ with a state of the form $(q, \varepsilon)$ for some $q \in Q$ is the target of a run of Strahler number 0. Runs whose target have a different state are easily seen to have Strahler number at most $s$. We conclude that the Strahler number of $\BB_s$ is at most $s$, which concludes the proof.
\end{proof}

\section{Application to $5$-BVAS and $2$-BVASS} \label{sec:5BVAS-algo-iteratif}

A $d$-BVAS (resp. a $d$-VAS) is a $d$-BVASS (resp. a $d$-VASS) with one state. By discarding this unique state, formally, a $d$-BVAS is given as a finite set of \emph{rules} $\Delta \subseteq \N\times\Z^d$ where a rule $\delta=(r,\vec{a})$ is defined by a natural number $r$ denoting the arity of $\delta$ and a vector $\vec{a}$ denoting the effect of $\delta$. Symmetrically, a $d$-VAS is a pair $(\Delta,\vec{S})$ where $\Delta$ is a finite subset of vectors in $\Z^d$ called \emph{actions} and $\vec{S}$ is a finite set of initial configurations in $\N^d$. Based on this observation, notions defined for $d$-BVASS and $d$-VASS are naturally propagated over $d$-BVAS and $d$-VAS.

\smallskip

Hopcroft and Pansiot proved in~\cite{HP79} that reachability sets of $5$-VAS are effectively semilinear while there exists a $3$-VASS with two states having a non-semilinear reachability set. Since reachability sets of $d$-VASS are sections of $d+3$-VAS reachability sets (see for instance~\cite[lemma 2.1]{HP79}), there also exists a $6$-VAS with a non-semilinear reachability set. In this section, we push those results on branching VAS and VASS by proving that reachability sets of $5$-BVAS and $2$-BVASS are effectively semilinear. The computation is performed by a fix-point algorithm. We only present the computation for $5$-BVAS since for $2$-BVASS the computation is similar (it is based on the same fix-point algorithm but instantiated with $2$-BVASS). Our computation for $5$-BVAS is based on an extended model of $5$-VAS, called $5$-LVAS that allow linear actions.

\smallskip

Formally, a \emph{$d$-dimensional linear VAS} (\emph{$d$-LVAS} for short) is a pair $\VV=(\Delta,\Gamma_{init})$ where $\Delta$ is a finite set of pairs $(\vec{z},\vec{V})$, called \emph{transitions}, where $\vec{z}\in\Z^d$ and $\vec{V}$ is a finite subset of $\N^d$, and $\Gamma_{init}$ is a presentation of a semilinear set of $\N^d$. We associate with a pair $\delta=(\vec{z},\vec{V})$ in $\Delta$ the set $\sem{\delta}=\vec{z}+\vec{P}$ where $\vec{P}$ is the periodic set spanned by $\vec{V}$. We also introduce $\sem{\Delta}=\bigcup_{\delta\in \Delta}\sem{\delta}$. A \emph{configuration} is a vector in $\N^d$ and a \emph{run} is a sequence $\rho=\vec{x}_0\vec{x}_1\ldots\vec{x}_k$ of configurations satisfying $\vec{x}_0\in\sem{\Gamma_{init}}$, and $\vec{x}_j-\vec{x}_{j-1}\in \sem{\Delta}$ for every $1\leq j\leq k$. Its last configuration is called \emph{target} and is denoted $\target(\rho)$. The set of runs of $\VV$ is denoted $\Runs(\VV)$. The \emph{reachability set} $\Reach(\VV)$ is the set of targets of runs of $\VV$, \emph{i.e.} $\Reach(\VV) = \{\target(\rho) \mid \rho \in \Runs(\VV) \}$.

\smallskip

Reachability sets of $d$-VAS are clearly reachability sets of $d$-LVAS. Conversely, reachability sets of $d$-LVAS are sections of reachability sets of $d$-VASS. In fact, any transition $\delta=(\vec{z},\vec{V})$ of a $d$-LVAS can be encoded in a $d$-VASS by two transitions $(q_0,\vec{z},r)$ and $(r,\vec{0},q_0)$ and loops $(r,\vec{v},r)$. Such a construction cannot be used for proving that reachability sets of $5$-LVAS are effectively semilinear since it introduces several control states. Nevertheless, in the following lemma, by decomposing runs of a $5$-LVAS with respect to the first time a transition $(\vec{z},\vec{V})$ with $\vec{V}\not=\emptyset$ is used, we provide a way to prove that reachability sets of $5$-LVAS are effectively semilinear.
\begin{lemma}\label{lem:poststar}
   Reachability sets of $5$-LVAS are effectively semilinear.
\end{lemma}
\begin{proof}
  A transition $\delta=(\vec{z},\vec{V})$ of a $5$-LVAS is said to be \emph{wide} if its set of periods $\vec{V}$ is non-empty. Let us first explain why reachability sets of $5$-LVAS without wide transition is effectively semilinear. To do so, we consider such a $5$-LVAS $\VV=(\Delta,\Gamma_{init})$. Observe that the reachability set of $\VV$ is equal to the finite union of the reachability sets of the $5$-LVAS $(\Delta,\{\gamma\})$ where $\gamma$ ranges over $\Gamma_{init}$. So, we can assume without loss of generality that $\Gamma_{init}$ is reduced to a single $\gamma=(\vec{x},\vec{V})$. Now, observe that the reachability set of such a $5$-LVAS is equal to the reachability set of the $5$-VAS $(\{\vec{b} \mid (\vec{b},\emptyset)\in \Delta\}\cup \vec{V},\{\vec{x}\})$. It follows that this reachability set is effectively semilinear from~\cite{HP79}.

  Now, let us prove by induction over $k\in\N$ that reachability sets of $5$-LVAS with $k$ wide transitions are effectively semilinear. The rank $k=0$ is proved in the previous paragraph. Assume the rank $k$ proved and let us consider a $5$-LVAS $\VV=(\Delta,\Gamma_{init})$ with $k+1$ wide transitions. We denote by $W$ the set of wide transitions and we introduce the $5$-LVAS $\VV_{init}$ obtained from $\VV$ by removing the wide transitions. From the previous paragraph, the reachability set of $\VV_{init}$ is effectively semilinear. It follows that we can compute a presentation $\Gamma_{init}^*$ of the reachability set of $\VV_{init}$. We associate with each wide transition $w=(\vec{z},\vec{S})$ of $W$, the set $\vec{S}_{w}=\N^d\cap (\sem{\Gamma_{init}^*}+\sem{w})$. From \cite{gs66}, this set is semilinear and we can compute a presentation $\Gamma_w$ of it. Now, let us consider the $5$-LVAS $\VV_w=(\Delta_w,\Gamma_w)$ obtained from $\VV$ by replacing the transition $w$ by $(\vec{z},\emptyset)$ and by replacing $\Gamma_{init}$ by $\Gamma_w$. By induction, we deduce that the reachability set of that $5$-LVAS is effectively semilinear. It follows that we can compute a presentation $\Gamma_w^*$ of that semilinear set. We are going to prove that $\Gamma^*$ defined as $\Gamma_{init}^*\cup \bigcup_{w\in W}\Gamma_w^*$ is a presentation of the reachability set of $\VV$.

  Clearly configurations in the semilinear set $\sem{\Gamma_{init}^*}$ are reachable for $\VV$. Let $w\in W$. Notice that configurations in the semilinear set $\sem{\Gamma_w}$ are reachable for $\VV$ since they can be obtained from configurations in $\sem{\Gamma_{init}^*}$ by executing the transition $w$. We deduce that configurations in the semilinear set $\Gamma_w^*$ are also reachable for $\VV$. It follows that configurations in $\sem{\Gamma^*}$ are reachable for $\VV$.

  Conversely, let us prove that reachable configurations of $\VV$ are in the semilinear set $\sem{\Gamma^*}$. So, let us consider a run $\rho=\vec{x}_0\vec{x}_1\ldots\vec{x}_k$ of $\VV$ and let us prove that $\vec{x}_k\in \sem{\Gamma^*}$. There exists a sequence $\delta_1,\ldots,\delta_k\in \Delta$ such that $\vec{x}_j-\vec{x}_{j-1}\in\sem{\delta_j}$ for every $1\leq j\leq k$. If $\delta_1,\ldots,\delta_k\in \Delta\setminus W$, we deduce that $\vec{x}_k\in \sem{\Gamma_{init}^*}$ and we are done. So, we can assume that there exists a minimal $i\in\{1,\ldots,k\}$ such that $\delta_i\in W$. By minimality of $i$, it follows that $\vec{x}_{i-1}\in\sem{\Gamma_{init}^*}$. We denote by $w=(\vec{z},\vec{V})$ the transition $\delta_i$ and by $\vec{P}$ the periodic set spanned by $\vec{V}$. We introduce the set $J$ of occurrences of $w$, i.e. $J=\{j\in\{i,\ldots,k\} \mid \delta_j=w\}$. For each $j\in J$, there exists $\vec{p}_j\in \vec{P}$ such that $\vec{x}_{j}=\vec{x}_{j-1}+\vec{z}+\vec{p}_j$. We introduce the sequence $\rho'=\vec{x}_i'\vec{x}'_{i+1}\ldots\vec{x}'_k$ where $\vec{x}_\ell'=\vec{x}_\ell+\sum_{j\in J\cap \{\ell+1,\ldots,k\}}\vec{p}_j$. Let us prove that $\rho'$ is a run of $\VV_w$.

  Since $\vec{x}_i=\vec{x}_{i-1}+\vec{z}+\vec{p}_i$, we deduce that $\vec{x}'_i=\vec{x}_{i-1}+\vec{z}+\sum_{j\in J}\vec{p}_j$. As $\vec{x}_{i-1}\in\sem{\Gamma_{init}}$, we get $\vec{x}'_i\in \sem{\Gamma_w}$.
  
  Now, let $\ell\in\{i+1,\ldots,k\}$ and let us prove that $\vec{x}_\ell'-\vec{x}_{\ell-1}'\in \sem{\Delta_W}$. If $\ell\not\in J$ this property is immediate since $\vec{x}_\ell'-\vec{x}_{\ell-1}'=\vec{x}_\ell-\vec{x}_{\ell-1}\in \sem{\delta_\ell}$. Since $\delta_\ell\in \Delta_w$, if follows that $\vec{x}_\ell'-\vec{x}_{\ell-1}'\in\sem{\Delta_w}$. If $\ell\in J$ then $\vec{x}_\ell'-\vec{x}_{\ell-1}'=\vec{x}_\ell-\vec{x}_{\ell-1}-\vec{p}_\ell=\vec{z}$. Since $(\vec{z},\emptyset)\in \Delta_w$, we deduce that $\vec{x}_\ell'-\vec{x}_{\ell-1}'\in\sem{\Delta_w}$ also in that case.

  We have proved that $\rho'$ is a run of $\VV_w$. Since $\target(\rho')=\vec{x}'_k=\vec{x}_k$ we have proved that $\vec{x}_k\in \sem{\Gamma_w}$.
\end{proof}

\newcommand{\instantiate}[2]{#1\left<#2\right>}
\newcommand{\poststar}{\textsc{poststar}}

We associate with a $d$-BVAS $\Delta$ and a presentation $\Gamma$ of a semilinear set $\vec{S}\subseteq\N^d$, the $d$-LVAS $\instantiate\Delta\Gamma=(T,\{\vec{a}\in\N^d \mid (0,\vec{a})\in\Delta\})$ where $T$ is the following set of transitions:
$$T=\bigcup_{(r,\vec{a})\in \Delta\mid r\geq 1}\{(\vec{a}+\sum_{\ell=1}^{r-1}\vec{b}_\ell,\bigcup_{\ell=1}^{r-1}\vec{V}_\ell) \mid (\vec{b}_1,\vec{V}_1),\ldots,(\vec{b}_{r-1},\vec{V}_{r-1})\in\Gamma\}$$
Clearly, if $\vec{S}\subseteq \Reach(\Delta)$ then $\Reach(\instantiate\Delta\Gamma)\subseteq \Reach(\Delta)$. In fact, from a run $\rho$ of $\instantiate\Delta\Gamma$ we obtain a run $\alpha$ of $\Delta$ with $\target(\rho)=\target(\alpha)$ just by considering $\rho$ as a (single branch) tree and by inserting as new children of nodes of that branch some runs $\alpha_{\vec{s}}$ of $\Delta$ such that $\target(\rho_{\vec{s}})=\vec{s}$ where $\vec{s}$ ranges over $\sem{\Gamma}$. 

\smallskip

Computation of $5$-BVAS reachability sets is based on the previous construction. We introduce a computable function $\poststar$ that maps a $5$-LVAS $\VV$ on a presentation $\poststar(\VV)$ of the semilinear set $\Reach(\VV)$. \cref{lem:poststar} shows that such a function exists. We introduce the sequence $(\Gamma_i)_{i\in\N}$ of presentations of semilinear sets defined by $\Gamma_0=\emptyset$ and by induction for every $i\in\N$ as follows:
$$\Gamma_{i+1}=\poststar(\instantiate\Delta{\Gamma_i})$$
Let $\vec{S}_i$ be the semilinear set presented by $\Gamma_i$. By induction on $i$ we deduce that $\vec{S}_i\subseteq \vec{S}_{i+1}$ and $\vec{S}_i\subseteq \Reach(\Delta)$. Notice that if $\vec{S}_{i+1}=\vec{S}_i$ then $\vec{S}_i$ is an inductive invariant for $\Delta$. It follows that $\vec{S}_i=\Reach(\Delta)$ for such an index $i$. As we can decide the inclusion of two semilinear sets given by presentations, we deduce that the reachability set of $5$-BVAS are effectively semilinear if such an index $i$ exists. We first prove the following lemma.
\begin{lemma}\label{lem:Si}
   $\vec{S}_i$ is the set of configurations reachable by a run of $\Delta$ of Strahler number at most $i$.
\end{lemma}
\begin{proof}
   We denote by $\VV_i$ the $5$-LVAS $\instantiate\Delta{\Gamma_i}$. Notice that $\VV_i=(T_i,\vec{S}_{init})$ for some finite set $T_i$ of actions and some finite set $\vec{S}_{init}=\{\vec{a}\in\N^d \mid (0,\vec{a})\in\Delta\}$ independent of $i$.

   We prove the lemma by induction on $i$. The rank $i=0$ is trivial. Assume the rank $i$ proved.

  Let us prove by induction on $k$ that every run $\vec{x}_0\ldots\vec{x}_k$ of $\instantiate\Delta{\Gamma_i}$ is such that $\vec{x}_k$ is a configuration reachable by a run of $\Delta$ of Strahler number at most $i+1$. When $k=0$ the proof is trivial since $\vec{x}_0\in \vec{S}_{init}$. Assume the rank $k$ proved and let us consider a run $\vec{x}_0\ldots\vec{x}_{k+1}$ of $\instantiate\Delta{\Gamma_i}$. By induction, we deduce that $\vec{x}_k$ is reachable by a run $\rho$ of $\Delta$ of Strahler number at most $i+1$. As $\vec{x}_{k+1}-\vec{x}_k\in \sem{T_i}$, there exists a rule $\delta=(r,\vec{a})\in\Delta$ and a sequence $(\vec{s}_{\ell})_{1\leq \ell< r}$ of configurations in $\vec{S}_i$ such that $\vec{x}_{k+1}-\vec{x}_{k}=\vec{a}+\sum_{\ell=1}^{r-1}\vec{s}_\ell$. By induction, there exists a run $\rho_\ell$ of $\Delta$ of Strahler number at most $i$ that reaches $\vec{s}_\ell$. Let us consider the tree $\rho'$ with a root labelled by the configuration $\vec{x}_{k+1}$ and with children $\rho_1,\ldots,\rho_{r-1},\rho$. Notice that $\rho'$ is a run of $\Delta$ of Strahler number at most $i+1$. We have proved the induction on $k$. It follows that every configuration in $\vec{S}_{i+1}$ is reachable by a run of $\Delta$ of Strahler number at most $i+1$.

  Now, let us prove by induction on $k$ that every configuration reachable by a run of $\Delta$ of Strahler number at most $i+1$ and of size $k$ is in $\vec{S}_{i+1}$. When $k=0$ the proof is trivial since there is no run of size $k$. Assume the rank $k$ proved an let us consider a run $\rho$ of $\Delta$ of Strahler number at most $i+1$ and of size $k+1$ that reaches a configuration $\vec{s}$. Let $\rho_1,\ldots,\rho_r$ be the runs of $\Delta$ corresponding to the children of the root of $\rho$, and let $\vec{s}_1,\ldots,\vec{s}_r$ the configurations reached by those runs. Since $\rho$ is a run, there exists $\vec{a}\in\Z^d$ such that $(r,\vec{a})\in\Delta$ and such that $\vec{s}=\vec{a}+\sum_{\ell=1}^r\vec{s}_\ell$. If $r=0$ then $\vec{s}\in\vec{S}_{init}$ and we deduce that $\vec{s}\in\vec{S}_{i+1}$. So, we can assume that $r\geq 1$. Since the Strahler number of $\rho$ is at most $i+1$, by reordering $\rho_1,\ldots,\rho_r$ we can assume without loss of generality that the Strahler number of $\rho_\ell$ is at most $i$ for every $1\leq \ell<r$. By induction on $i$, we deduce that the configuration $\vec{s}_\ell$ is in $\vec{S}_i$. Moreover, as the Strahler number $\rho_r$ is at most $i+1$ and the size of $\rho_r$ is at most $k$, by induction on $k$, we deduce that $\vec{s}_r\in\vec{S}_{i+1}$. As $\vec{s}_\ell\in \vec{S}_i$, there exists $\gamma_\ell=(\vec{b}_\ell,\vec{V}_\ell)\in \Gamma_i$ such that $\vec{s}_\ell\in\sem{\gamma_\ell}$. Let us consider $t=(\vec{a}+\sum_{\ell=1}^{r-1}\vec{b}_\ell,\bigcup_{\ell=1}^{r-1}\vec{V}_\ell)$. Observe that $t$ is a transition in $T_i$. Moreover, $t$ can be executed from $\vec{s}_r$ and provides $\vec{s}$. From $\vec{s}_r\in\vec{S}_{i+1}$ we get $\vec{s}\in\vec{S}_{i+1}$. We have proved the induction on $k$. It follows that that every configuration reachable by a run of $\Delta$ of Strahler number at most $i+1$ is in $\vec{S}_{i+1}$.

  So, we have proved the induction on $i$, i.e. we have proved the lemma.
\end{proof}

\cref{cor:bounded-Strahler} shows that there exists $s\in \N$ such that every configuration reached by $\Delta$ is reached by a run whose Strahler number is at most $s$. From the previous lemma, we deduce that for every $i\geq s$ we have $\vec{S}_{i+1}=\vec{S}_i$. We have proved the following theorem.
\begin{theorem}
  Reachability sets of $5$-BVAS and $2$-BVASS are effectively semilinear.
\end{theorem}

\section{Conclusion} \label{sec:conclusion}
In this paper, we showed how to bridge the gap between VASS and branching VASS. In small dimensions, we generalized the effective semilinearity of reachability sets of $2$-dimensional VASS and $5$-dimensional VAS to the branching case, via a simple iterative fix-point algorithm. We obtained those results by applying techniques that are valid in arbitrary dimension.

We introduced a well-quasi-order (wqo) on the set of runs of BVASS that generalizes the well-known wqo on VASS runs and that satisfies, as in the case of VASS, an \emph{amalgamation property}. Using this property, we have proved that every BVASS can be extended with an appropriate finite set $\vec F$ of reachable configurations to get an equivalent BVASS with bounded branching depth complexity. This result has three consequences. First, reachability sets of BVASS are VASS sections, and hence, are almost semilinear. Second, if the reachability set of a BVASS is semilinear then it is effectively computable. Third, the Strahler number of a BVASS is finite, but not computable. Whether such an appropriate set $\vec F$ is computable is an open problem. Similarly, we do not know if an upper bound on the Strahler number is computable. Finally, we proved that the Strahler-bounded reachability problem is decidable, hence the computability of such an upper-bound would entail the decidability of the reachability problem for BVASS. The latter problem is still open, but hopefully not for a long time!

\bibliographystyle{alphaurl}
\bibliography{biblio}

\end{document}